\documentclass[a4paper,10pt, draftclsnofoot, onecolumn]{IEEEtran}
\usepackage[normalem]{ulem}
\usepackage{multirow}
\usepackage{caption}
\usepackage{subcaption}
\usepackage{caption}
\usepackage{graphicx}
\usepackage{longtable}
\usepackage{courier}
\usepackage{float}
\usepackage{cite}
\usepackage{amsmath,amssymb}
\usepackage{bbold}
\usepackage{dsfont}
\usepackage[in]{fullpage}
\usepackage[verbose]{wrapfig}
\usepackage{amsthm}
\usepackage{authblk}
\usepackage{xcolor}

\DeclareMathOperator*{\eq}{=}

\newcommand{\dfn}{\stackrel{\triangle}{=}}

\usepackage[T1]{fontenc}
\usepackage[utf8]{inputenc}
\usepackage{authblk}
\usepackage[margin=0.7in]{geometry}
\newtheorem{theorem}{Theorem}

\title{Converse Bounds on Modulation-Estimation Performance for the Gaussian Multiple-Access Channel \footnote{This work was presented in part at the 2016 IEEE International Symposium on Information Theory (ISIT 2016), Barcelona, Spain, July 10-15, 2016}}

\author[1]{Ay\c{s}e \"{U}nsal}
\author[2]{Raymond Knopp}
\author[3]{Neri Merhav}
\affil[1]{Univ Lyon, INSA Lyon, Inria, CITI, France, {\tt \small{ayse.unsal@insa-lyon.fr}}} 
\affil[2]{Communications Systems Department, Eurecom, 
France, {\tt \small{raymond.knopp@eurecom.fr}} }
\affil[3]{The Andrew and Erna Viterbi Faculty of Electrical Engineering, Technion, Israel, {\tt \small{merhav@ee.technion.ac.il}}} 
 
\begin{document}

\maketitle
\thispagestyle{plain}
\begin{abstract}
  This paper focuses on the problem of separately modulating and jointly estimating two independent continuous-valued parameters sent over 
a Gaussian multiple-access channel (MAC) under the mean square error (MSE) criterion. 
To this end, we first improve an existing lower bound on the MSE that is obtained using the parameter modulation-estimation techniques for the single-user additive white Gaussian noise (AWGN) channel.
As for the main contribution of this work, this improved modulation-estimation analysis is generalized to the model of the two-user Gaussian MAC, which will likely become an important mathematical framework for the analysis of remote sensing problems in wireless networks. We present
outer bounds to the achievable region in the plane of the MSE's of the two user
parameters, which provides a trade-off between the MSE's, in addition to the upper bounds on the achievable region of the MSE exponents, namely, the exponential decay rates of these MSE's in the asymptotic regime of long blocks.
\end{abstract}

\pagestyle{plain}
\section{Introduction \label{sec:intro}}

Before addressing the problem of joint modulation-estimation for the Gaussian MAC, let us refer first to
the more fundamental single-user modulation-estimation problem. In this setting, a single continuous--valued random parameter $U$ is encoded (modulated) into an $N$-dimensional power-limited vector $\mathbf{x}(U)$ and 
transmitted over an additive-white Gaussian noise (AWGN) channel
\cite{goblick,WynerZiv69,WozencraftJacobs} as shown in Fig. \ref{fig:model}(a). The corresponding $N$-dimensional channel output 
vector is given by $\mathbf{y}=\mathbf{x}(U)+\mathbf{z}$, where $\mathbf{z}$ is a
Gaussian noise vector with independent and identically distributed (i.i.d.) components, which are independent also of $U$. The channel output vector $\mathbf{y}$ is used by the receiver to estimate $U$ by an estimator $\hat{U}(\mathbf{y})$.
The goal is to derive a lower bound to the
MSE, $\mathbf{E}(U-\hat{U}(\mathbf{y}))^2$, that applies to every modulator
$\mathbf{x}(\cdot)$, that is subjected to a given power constraint, and to every
estimator $\hat{U}(\mathbf{y})$ \cite[Chapter 8]{WozencraftJacobs}. 
More recently in \cite{Merhav12b},
this class of transmission problems was given 
the name {\em parameter modulation-estimation}, which we believe, will likely 
become an important mathematical framework to analyze various 
remote sensing problems that may arise in fifth generation 
wireless networks. 
\begin{figure}
\centering
\includegraphics[width=0.7\linewidth]{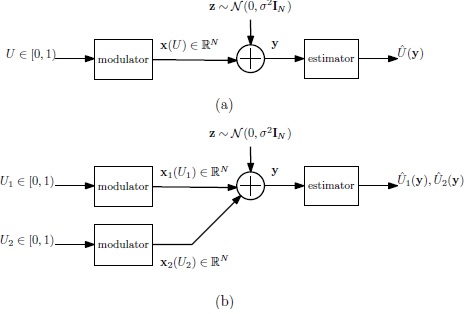}
\caption{System Models \label{fig:model}}
\end{figure}
The purpose of this work is to extend the described problem, as well as its analysis and results,
to the model of the discrete-time two-user Gaussian MAC, where two independent 
parameters, denoted by $U_1$ and $U_2$, are conveyed from two separate transmitters and
jointly estimated at the receiver. This model is shown in Fig. \ref{fig:model}(b).
The aim is to derive outer bounds on the region of best achievable MSE's associated with any modulators (subjected to power constraints) and estimators of these parameters. It should be noted that in the context of the MAC model considered here, there exists an interesting trade--off that is not seen in the single--user case described in the first paragraph above.
A better modulator for one of the users is good, of course, for the estimation of the corresponding parameter at the receiver side, because it amounts to high sensitivity of the likelihood function to this parameter. However, at the same time, and for the very same reason, it comes at the expense of the estimation performance of the other user (for which the parameter of the first user is a nuisance parameter). Indeed, such a trade--off is manifested in the boundary curves of the achievable regions that we obtain, which are always monotonically non-increasing functions, namely, smaller MSE values in one parameter impose higher lower bounds on the MSE values of the other. This paper builds on relationships between modulation and coding and between estimation and detection.

The remote-sensing application is one where the random-variables $U_i$ are measured by a communicating device equipped with some form of analog sensor.  The resulting measurements are conveyed to the network via the uplink of a wireless communication system. In the near future such devices will use conventional cellular access, albeit with specially-tailored waveforms, to feed data centers with physical information observed in so-called {\em smart cities} or remote areas. These applications will often impose extremely low-periodicity sporadic transmission coupled with long lifetime batteries or solar cells in order to remain embedded in nature with little or no maintenance for long periods of time. In addition, the problem addressed here is also related to more general ranging estimation problems where the random parameters are induced by the channel.  As an example, consider a satellite or cellular positioning system where the $U_i$ represent two time-delays which, when estimated at the receiver, are used to estimate the position of the receiver. The framework considered here can therefore be extended to analyze the fundamental performance limits in such systems. 
\subsection{Related Work}

The majority of work dealing with this class of problems considers transmission 
on a continuous-time channel using finite-energy waveforms without bandwidth constraints. In \cite{goblick}, Goblick provided a 
lower bound of the exponential order of $\exp\left(-2\mathcal{E}/N_0\right)$, where 
$\mathcal{E}$ is the energy used to convey $U$
and $N_0/2$ is the two-sided power spectral density of the channel noise process. Goblick also provided several examples of parameter modulation-estimation schemes, 
one of them turns out to achieve the best asymptotic performance, namely, MSE of the exponential order of $\exp\left(-\mathcal{E}/3N_0\right)$. This is a simple digital scheme, which is based on first uniformly quantizing the parameter into one out of $M$ points and then transmitting the index of the quantized parameter to the receiver, using $M$-ary orthogonal modulation scheme. Another modulation strategy, considered this problem in continuous-time,
was given in \cite[pp.\ 623]{WozencraftJacobs} where the parameter is reflected in the delay of a 
purely analog signaling pulse sent across the channel, namely, pulse position modulation (PPM). 
When the pulse bandwidth is unlimited, this system achieves the same exponential behaviour as Goblick's scheme. 
This scheme also provided a link to the classical ranging problem 
where the objective is to estimate the random delay of an incoming waveform 
corrupted by Gaussian noise \cite{ZivZakai}. 
In \cite{WynerZiv69}, Wyner and Ziv showed that Goblick's 
lower bound could be improved to the order of $\exp\left(-\mathcal{E}/2N_0\right)$. 
Cohn \cite{Cohn} and Burnashev \cite{Burnashev79}, \cite{Burnashev84}, \cite{Burnashev85}, further improved the multiplicative factor at the MSE exponent, progressively from 1/2.889 to 1/2.896. then 1/2.970, and finally to 1/3.000, thus closing the gap to Goblick's 
practical scheme. In particular, despite the significance of the presented results, unfortunately, \cite{Cohn} is not well known
as it has never been published and hence is not easily accessible to the general public. In a nutshell, in \cite{Cohn} Cohn presented lower bounds on the average MSE in estimating the message of a single user 
using a geometric approach for simplex signal sets as well as the general case. 
The main contribution of \cite{Merhav12b} was the characterization of 
the parameter modulation-estimation problem for infinite-dimensional transmission over the continuous-time AWGN channel. 
A recent example of a similar scenario as the present paper can be found in \cite{Unsal-CISS, Unsal-thesis}, 
where lower bounds on the MSE region are provided for 
the transmission of two correlated analog source samples with and without causal feedback on 
the discrete-time AWGN MAC without a constraint on the number of signal dimensionality. The main difference between the current paper and \cite{Unsal-CISS, Unsal-thesis} is the analysis technique that is used. \cite{Unsal-CISS, Unsal-thesis} use an information--theoretic approach to obtain lower bounds.

\subsection{Contributions}

This paper studies the problem of jointly modulating and estimating two independent continuous-valued random variables encoded into an $N$--dimensional vector and transmitted over an AWGN channel to be estimated at the receiver end. The performance criterion is chosen as the MSE, which is characterized in two different ways as follows. Firstly, we derive outer bounds on the achievable region of pairs $(\mathrm{MSE}_1,\mathrm{MSE}_2)$, where $\mathrm{MSE}_1$ and $\mathrm{MSE}_2$ are the MSE's associated with arbitrary parameters, using a generalization of Shannon's zero-rate lower bound \cite{Shannon59} for the two-user discrete-time MAC, which allows us to characterize the MSE region in terms of the signal--to--noise ratios. We present outer bounds to the achievable region in the plane of the MSE's, basically one MSE associated to one of the users is bounded by a function that depends on the MSE associated to the other user. Thus, we obtain a trade-off between the MSE's based on some parameter.

In addition, we investigate the exponential behaviour of $(\mathrm{MSE}_1,\mathrm{MSE}_2)$ by characterizing a lower bound to the region of achievable pairs of MSE exponents for any joint parameter-modulation estimation scheme. To this end,
we adapt the multiple-access results of \cite{Nazari-thesis} to the discrete-time AWGN channel. In order to find the tightest characterization, we also use the bounds on the on the reliability function of the Gaussian channel proposed in \cite{Shannon59, Ashikhmin}. Coupled with the results of \cite{Weng}, we provide the means to make use of single-user error exponents for the characterization of multiuser channels.

\subsection{Outline}
In Section \ref{sec:model}, we describe the system model and formalize the problem. In Section \ref{sec:single_user}, we begin with the single-user case and present lower bounds on the MSE itself and its MSE exponent, as a preparatory step to be used later in the MAC model.
Section \ref{sec:MAC} is focused on the generalization of parameter modulation-estimation problem to a two-user Gaussian MAC in two subsections. In Subsections \ref{subsec:dual_zr_linear} and \ref{sec:nonzero}, respectively, we present new lower bounds on the MSE's and the MSE exponents.
The proposed bounds are numerically compared in Section \ref{sec:numerical}. Finally, in Section \ref{sec:conc}, we draw conclusions from our results.

\section{Problem Formulation and Signal Models \label{sec:model}}

\subsection{Single-user setting \label{sec:model_1}}
We consider lower bounds on the MSE of modulation-estimation schemes for a random parameter $U$, that is uniformly distributed over the interval $[0,1)$. 
\footnote{The results presented in this paper can be quite easily adapted to other source distributions.}
The parameter $U$ is conveyed by a modulator, which maps $U$ into a channel input vector $\mathbf{x}(U)$ that is transmitted over an $N$-dimensional memoryless AWGN channel, which is assumed to be phase-synchronous. In general, we have the following signal model
\begin{equation}\label{eq:AWGN_single}
 \mathbf{y} = \mathbf{x}(U) + \mathbf{z} 
\end{equation} where $\mathbf{x}(U) $ is constrained in energy as 
\begin{equation} \label{eq:en_const}
\|\mathbf{x}(U)\|^2\leq N\mathcal{S} = \mathcal{E},
\end{equation} $\mathcal{S}$ and $\mathcal{E}$ being the power and energy limitations, respectively, and the noise covariance matrix is given by
\begin{equation}
\mathbf{E}\mathbf{z}\mathbf{z}^T=\sigma^2 \mathbf{I}_N.
\end{equation} Here the superscript $T$ denotes the transposition of a vector and $\mathbf{I}_N$ is the $N\times N$ identity matrix. 
At the receiver, we consider an estimator $\hat{U}(\mathbf{y})$ with corresponding $\mathrm{MSE}_{\mathrm{s}}={\mathbf{E}[U-\hat{U}(\mathbf{y})]^2}$. Let us also define the {\it asymptotic MSE exponent} as
\begin{equation} \label{eq:MSE_expo_1}
\epsilon_{\mathrm{s}} \dfn -\liminf_{N\to\infty}\frac{1}{N}\log \mathbf{E}[\hat{U}(\mathbf{y})-U]^2.
\end{equation} 

\subsection{Two-user setting \label{sec:model_2}}

For this setting, we generalize the model of eq. (\ref{eq:AWGN_single})
to a model that includes two independent random variables, $U_1$ and $U_2$, both uniformly distributed over $[0,1)$. These two parameters are separately conveyed by the modulators of two different users, which generate the channel input vectors $\mathbf{x}_1(U_1)$ and $\mathbf{x}_2(U_2)$ over an $N$-dimensional real-valued AWGN MAC obeying the following signal model
\begin{equation} \label{eq:MAC_def}
\mathbf{y} =\mathbf{x}_1(U_1) + \mathbf{x}_2(U_2) +\mathbf{z}.
\end{equation} 
The modulators are constrained in energy as
\begin{equation} \label{eq:enr}
\|\mathbf{x}_j(U_j)\|^2\leq N\mathcal{S}_j = \mathcal{E}_j, \;\forall U_j, \; \textrm{for}\;\;j=1,2
\end{equation}
and the noise covariance matrix is as before. As in the single--user case of Subsection \ref{sec:model_1}, at the receiver, we consider estimators $\hat{U}_j(\mathbf{y})$ with MSE's, $\mathrm{MSE}_j=\mathbf{E}[U_j-\hat{U}_j(\mathbf{y})]^2$, $j=1,2$. 
As mentioned earlier, in Section \ref{sec:MAC}, we derive outer bounds to the region of achievable MSE pairs $(\mathrm{MSE}_1,\mathrm{MSE}_2)$, which apply to arbitrary modulators and estimators subject to the aforementioned power limitations, $\mathcal{S}_1$ and $\mathcal{S}_2$. 
The first characterization is for a given finite $N$ and it provides a direct characterization of  $(\mathrm{MSE}_1,\mathrm{MSE}_2)$, whereas the second characterization is asymptotic and it characterizes the region in terms of the exponents $(\epsilon_1,\epsilon_2)$ where 
\begin{equation} \label{eq:expo_mac}
\epsilon_j\dfn -\liminf_{N\to\infty}\frac{1}{N}\log \mathbf{E}[\hat{U}_j(\mathbf{y})-U_j]^2 , \; \textrm{for}\;j=1,2.
\end{equation}
\section{Single--User Channel \label{sec:single_user}}
In this section, we first recall the single-user approach from \cite{Merhav12b} and improve the lower bound on the MSE for any parameter-modulator scheme.
Additionally, we present a new bound on the MSE exponent of a single--user channel.

\subsection{An improved  lower bound \label{subsec:zr}}

It is shown in \cite[eq.\ (21)]{Merhav12b} that for the single-user problem, the probability that the absolute  
estimation error ${|\hat{U}(\mathbf{y})-U|}$ would exceed $\Delta/2$, for a given $\Delta > 0$, is lower bounded as follows
\begin{equation} \label{eq:LB}
\Pr \{|\hat{U}(\mathbf{y})-U|>\Delta/2\}\geq L_B(\Delta)
\end{equation} where $L_B(\Delta)$ designates a lower bound to be specified later.
To derive such a bound, one considers the following hypothesis testing problem with $M$ equiprobable hypotheses,
\begin{equation}\label{eq:hypo_single}
 \mathcal{H}_{i} : \mathbf{y} = \mathbf{x}(u+i{\Delta}) + \mathbf{z}, 
 \end{equation} for $i \in \{1,\cdots,M\}$ where $u$ is considered a parameter taking values in $\left[0,1-(M-1)\Delta \right)$.
The lower bound $L_B(\Delta)$ is derived by combining the Ziv-Zakai approach with any lower bound on the average probability of error of an arbitrary code at a given rate. Specifically, let $\hat{i}$ denote the maximum likelihood (ML) estimate of $i$, and let $P_{\mbox{\tiny e}}(u,\Delta)=\Pr\left(\hat{i}\neq i| u\right)$ denote the corresponding conditional probability of error, which is upper bounded as follows:
\begin{align}
\int_{0}^{1-(M-1)\Delta}du \cdot P_e(u,\Delta)
&\leq\frac{1}{M}\sum_{i=0}^{M-1}\int_{0}^{1-(M-1)\Delta}\mathrm{d}u\cdot\Pr\left\{|\hat{U}(\mathbf{y})-U|>\Delta/2{\Bigm |}U=u+i\Delta\right\}\notag\\
&=\frac{1}{M}\sum_{i=0}^{M-1}\int_{i\Delta}^{1-(M-1)\Delta+i\Delta}\mathrm{d}u\cdot\Pr\left\{|\hat{U}(\mathbf{y})-U|>\frac{\Delta}{2}{\Bigm |}U=u\right\}\notag\\
&\overset{(a)}{=}\frac{1}{M}\sum_{i=0}^{M-1}\Pr\left\{|\hat{U}(\mathbf{y})-U|>\Delta/2,i\Delta\leq U\leq 1-(M-1)\Delta+i\Delta\right\}\nonumber \\
&\leq \frac{1}{M}\Pr\left\{|\hat{U}(\mathbf{y})-U|>\Delta/2\right \} \label{eq:Merhav13mod}
\end{align}
We note that (\ref{eq:Merhav13mod}) is valid for all $M$ and $\Delta$ such that $(M-1)\Delta<1$. If we add the condition that $M\Delta>1$, which amounts to $1/\Delta<M<1+(1/\Delta)$ or equivalently $M=\left\lceil1/\Delta\right\rceil$, the intervals in step (a) become disjoint. This yields
\begin{equation}
\frac{1}{\left\lceil 1/\Delta \right\rceil}\sum_{i=0}^{\left\lceil 1/\Delta \right\rceil-1}\Pr\left\{|\hat{U}-U|>\frac{\Delta}{2},i\Delta\leq U\leq 1-(\left\lceil 1/\Delta\right\rceil-1)\Delta+i\Delta\right\}\leq \frac{1}{\left\lceil 1/\Delta\right\rceil}\Pr\left\{|\hat{U}(\mathbf{y})-U|>\frac{\Delta}{2}\right\} 
\end{equation}
Bounding the left hand side (l.h.s.) of (\ref{eq:Merhav13mod}) using any zero-rate bound for $M$-ary signals, $P_\mathrm{ZR}\left(\mathcal{E},\left\lceil\frac{1}{\Delta}\right\rceil\right)$ yields the bound
\begin{equation} \label{eq:def_LB}
 \left\lceil 1/\Delta\right\rceil(1+\Delta-\left\lceil 1/\Delta \right\rceil\Delta)\cdot P_\mathrm{ZR}\left(\mathcal{E},\left\lceil 1/\Delta\right\rceil\right) \leq \Pr\left\{|\hat{U}(\mathbf{y})-U|>\Delta/2\right\}
\end{equation}
which is $M$ times larger than the original result given by \cite[eq. (21)]{Merhav12b}. The lower bound $L_B(\Delta)$ corresponds to the l.h.s. of (\ref{eq:def_LB}).
The right hand side of the last inequality is related to the MSE according to
\begin{align} \label{eq:mse_derv}
&\int_{0}^{1} d \Delta \cdot \Delta \cdot \Pr \{|\hat{U}(\mathbf{y})-U|>\Delta/2\}\nonumber \\ 
&\overset{(a)}{\leq}4 \int_{0}^{1} d \delta \cdot \delta \cdot \Pr \{|\hat{U}(\mathbf{y})-U|>\delta\}\overset{(b)}{=} 2\mathbf{E} [\hat{U}(\mathbf{y})-U]^2
\end{align} where in (a), we changed the integration variable to $\delta=\Delta/2$ and the integration 
interval was extended to $[0,1)$, whereas in (b), the following identity was used
\begin{equation} \label{eq:MSE_def}
\mathbf{E}[\hat{U}(\mathbf{y})-U]^2 =2\int_{0}^{1}d \Delta \cdot \Delta \cdot \Pr \{|\hat{U}(\mathbf{y})-U|>\Delta\}. 
\end{equation}
Combining (\ref{eq:Merhav13mod}) with (\ref{eq:mse_derv}), the improved single-user lower bound is given by
\begin{align} \label{merhav-mod-lhs}
\mathrm{MSE}_{\mathrm{s}} &\geq \frac{1}{2} \int_{0}^{1}d \Delta \left\lceil 1/\Delta\right\rceil \Delta  \left(1+ \Delta- \Delta \left\lceil 1/\Delta\right\rceil\right) P_\mathrm{ZR}\left(\mathcal{E},\left\lceil 1/\Delta\right\rceil\right) \nonumber
 \\
&=\frac{1}{2} \sum_{i=2}^{\infty}\int_{1/i}^{1/(i-1)}d \Delta \cdot \left(  \Delta i+ \Delta^2 i- \Delta^2 i^2\right) \cdot P_\mathrm{ZR}\left(\mathcal{E},i\right) \notag \\
&= \frac{1}{2} \sum_{i=2}^{\infty} \frac{3i-2}{6i^2(i-1)^2} P_\mathrm{ZR}\left(\mathcal{E},i\right).
 \end{align} 

In what follows we consider two zero-rate bounds.
\subsubsection{Shannon zero-rate bound \cite{Shannon59}}
In \cite[eq. (81)]{Shannon59} we have the general zero-rate lower bound 
\begin{equation} \label{eq:Sh_single_zero}
P_\mathrm{ZR}^{\mathrm{Shannon}}\left(\mathcal{E},M\right)\triangleq\frac{1}{M}\sum_{m=2}^{M}Q\left(\sqrt{\frac{m}{m-1}\left(\frac{\mathcal{E}}{2 \sigma^2}\right)}\right)
\end{equation}
which is valid for all $N$ and can be used in conjunction with (\ref{merhav-mod-lhs}) to bound the MSE for a point-to-point AWGN channel.

\subsubsection{A new zero-rate lower bound \label{subsec:poli}}

Using the Polyanskiy {\em et al.} converse \cite[Theorem 41]{Polyanskiy} for the AWGN channel which provides a lower bound on the average error probability for any $M$-ary signal set in $N$-dimensions, we propose a new lower bound on the error-probability for $N\rightarrow\infty$ under the finite-energy constraint in (\ref{eq:en_const}) given as
\begin{equation} \label{eq:Poly_single_zero}
P_\mathrm{ZR}^{\mathrm{P}}\left(\mathcal{E},M\right)\triangleq Q\left(\frac{\sqrt{\mathcal{E}}}{\sigma}(1+\mu)-Q^{-1}\left(\frac{1}{M}\right)\right)
\end{equation}
for any arbitrarily small $\mu>0$. The derivation of $P_\mathrm{ZR}^{\mathrm{P}}\left(\mathcal{E},M\right)$ can be found in detail in Appendix \ref{subsec:App_poli}. The expression in (\ref{eq:Poly_single_zero}) is potentially tighter than (\ref{eq:Sh_single_zero}) for low signal energies since it increases to 1 with $M$ for a fixed energy as is the case for any real signal set. It is clearly looser asymptotically since the energy exponent for fixed $M$ is $\mathcal{E}/2\sigma^2$ and not $\mathcal{E}/4\sigma^2$. We show a comparison of (\ref{eq:Poly_single_zero}) and (\ref{eq:Sh_single_zero}) with the error probability of the simplex signal set for $M=256$ in Figure \ref{fig:newbound}. The latter is widely believed to be the optimal signal set for $M$-ary equal-energy signals.  We see that (\ref{eq:Poly_single_zero}) is much closer to the Simplex error-probability for low signal-energies (error probabilities below $10^{-2}$) and crosses the Shannon bound at an error-probability around $10^{-10}$.   

\begin{figure}
\centering
\includegraphics[width=0.6\linewidth]{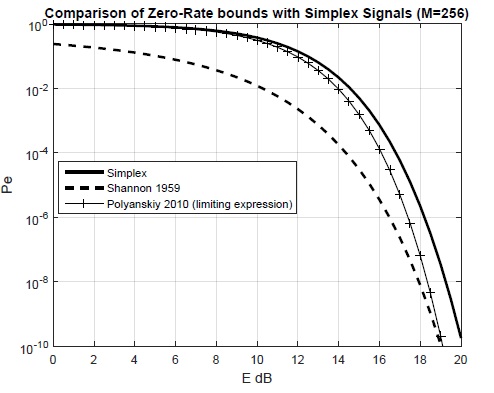}
\caption{Comparison of Zero-Rate Bounds with the error-probability of a Simplex ($M=256$)\label{fig:newbound}}
\end{figure}
\subsection{Upper bound on the MSE exponent  \label{subsec:non_zero_single}}

In this subsection, we introduce a new bound on the MSE exponent $\epsilon_{\mathrm{s}}$ defined by (\ref{eq:MSE_expo_1}) that makes use of any upper bound on the error exponent in a single user AWGN channel.
\begin{theorem} \label{t_single_nonzero}
For an arbitrary $N$-dimensional modulator $\mathbf{x}(U)$ subject to a power constraint given by (\ref{eq:en_const}) for transmission over the AWGN channel defined (\ref{eq:AWGN_single}) and for $R\geq 0$, the MSE exponent $\epsilon_{\mathrm{s}}$ as defined in (\ref{eq:MSE_expo_1}) is bounded by 
\begin{equation}\label{eq:expo_1}
\epsilon_{\mathrm{s}} \leq \min_R[2R+E_{\mbox{\tiny u}}(R)],
\end{equation}
where $E_{\mbox{\tiny u}}(R)$ is any upper bound on the error exponent function of the single user Gaussian channel.
\end{theorem}
\begin{proof}
Let us select $\Delta= e^{-RN}$ where $R\geq 0$ is a parameter (to be chosen later) in the general form of the bound 
\begin{equation}\label{eq:transform}
\mathbf{E}[\hat{U}(\mathbf{y})-U]^2\geq 2\int_{0}^{1}d \Delta \cdot \Delta \cdot L_B(\Delta),
\end{equation} where $L_B(\Delta)$ is the l.h.s. of (\ref{eq:def_LB}).
Changing the integration variable on the right-hand side (r.h.s.) of (\ref{eq:transform}) to $R$, we obtain
\begin{equation} \label{eq:single_nonzero}
\mathbf{E}[\hat{U}(\mathbf{y})-U]^2\geq \frac{N}{2} \int_{0}^{\infty} dR \cdot e^{-2RN}\cdot L_B(e^{-RN})
\end{equation}  
The r.h.s.\ of (\ref{eq:single_nonzero}) is bounded by an expression of the exponential order of $\smash{\exp \{-N \min_R[2R+E_{\mbox{\tiny u}}(R)]\}}= \mathrm{e}^{-N F}$ where $F \dfn \min_R[2R+E_{\mbox{\tiny u}}(R)]$. 
Finally, by taking the logarithms of both sides of (\ref{eq:single_nonzero}), dividing by $-N$, and passing to the limit $N\to\infty$, the proof of Theorem \ref{t_single_nonzero} is completed.
\end{proof}

As for an upper bound on the error exponent, $E_{\mbox{\tiny u}}(R)$, of the Gaussian channel, there are many options in the literature, such as Shannon's sphere-packing bound on the reliability function of the Gaussian channel \cite{Shannon59}, or a more recent bound by  Ashikhmin {\it et al.} \cite{Ashikhmin}, or others such as \cite{Burnashev} and \cite{Litsyn}. In this paper, we will use the results of \cite{Shannon59} and \cite{Ashikhmin} in our numerical evaluations due to their lower computational complexity relative to the others.

\subsubsection{Sphere-packing bound \label{subsubsec:sphere_packing}}
For rates confined to $[0,\mathcal{C})$, where $\mathcal{C}=(1/2) \log(1+A)$ is the Gaussian channel capacity, $A=\mathcal{S}/\sigma^2$ being the signal-to-noise ratio (SNR), Shannon's sphere-packing bound $E_{\mbox{\tiny sp}}(\psi(R),A)$ is an upper bound on the reliability function of the Gaussian channel $E(R,A)$ \cite{Shannon59} where $\psi(R)= \arcsin(e^{-R})$. The sphere-packing bound is given by
\begin{align}\label{eq:Shannon_rel}
E_{\mbox{\tiny sp}}(\psi(R),A)&=\frac{A}{2}-\frac{A(1-e^{-2R})}{4}+\frac{\sqrt{A(1-e^{-2R})(A(1-e^{-2R})+4)}}{4} +R \nonumber \\
&+\log 2 -\log \left( \sqrt{A(1-e^{-2R})}+\sqrt{A(1-e^{-2R})+4}\right)
\end{align}
The only positive and real minimizer of $E_{\mbox{\tiny sp}}(\psi(R),A)+2R$ where $E_{\mbox{\tiny sp}}(\psi(R),A)$ is given by (\ref{eq:Shannon_rel}) is obtained as 
\begin{equation}\label{eq:Rmin_sh}
R_{\min}=\frac{1}{2}\log \left \{\frac{A+\sqrt{A^2-2A+9}+3}{6}\right \}.
\end{equation} 

\subsubsection{Upper Bound by Ashikhmin {\it et al.}}
As for the second alternative to be used for $E_{\mbox{\tiny u}}(R)$ we have a more recent result by Ashikhmin {\it et al.} \cite[Theorem 1]{Ashikhmin}, which states that $E(R,A)\leq E_{\mbox{\tiny ABL}}(R,A)$, with $E_{\mbox{\tiny ABL}}(R,A)$ being defined as
\begin{equation}  \label{eq:Ashikhmin_rel}
E_{\mbox{\tiny ABL}}(R,A)=\underset{0\leq \rho \leq \rho_{k,l}}{\mathrm{min}} \underset{w,d}{\mathrm{max}} \left [\min \left( Ad^2/8, Aw^2/8- L_{\mbox{\tiny ABL}}(w,d,\rho) \right) \right]
\end{equation} where $0\leq d \leq d_{\max}$ and $d\leq w \leq w_{\max}$ with 
$$d_{\max}=\frac{\sqrt{2}(\sqrt{1+\rho_{kl}}-\sqrt{\rho_{kl}})}{\sqrt{1+2\rho_{kl}}}$$
and 
$$w_{\max}=\frac{\sqrt{2}(\sqrt{1+\rho}-\sqrt{\rho})}{\sqrt{1+2\rho}},$$ respectively. $\rho_{kl}$ is the root of the equality $$R-(1+\rho)H(\rho/(1+\rho))=0.$$ Here $H(x)$ denotes the binary entropy function. Lastly, for the inner minimization function of the bound $E_{\mbox{\tiny ABL}}(R,A)$, $L_{\mbox{\tiny ABL}}(w,d,\rho)$ is given by
\begin{equation}
L_{\mbox{\tiny ABL}}(w,d,\rho)=\min \left \{ \frac{Ad^2w^2}{8(4w^2-d^2)},F_{\mbox{\tiny ABL}}(1-w^2/2,\rho)\right \}
\end{equation}
with 
\begin{multline}
F_{\mbox{\tiny ABL}}(x,\rho)=R-(1+\rho)H(\rho/(1+\rho)) +\log ((x+\sqrt{(1+2\rho)^2x^2-4\rho(1+\rho)})/2)
\\
-(1+2\rho)\log \left(\frac{(1+2\rho)x+\sqrt{(1+2\rho)^2x^2-4\rho(1+\rho)}}{2(1+\rho)}\right).
\end{multline}
In Section \ref{sec:MAC}, the relation of these bounds with the two-user setting are analyzed and in Section \ref{sec:numerical}, their performances are numerically compared. It is worth mentioning that, unlike Shannon's results, the rate that minimizes $E_{\mbox{\tiny ABL}}(R,A)$ cannot be derived analytically. 
\section{Multiple-Access Channel \label{sec:MAC}}

In order to derive outer bounds for the two--user modulation--estimation problem, we consider the following auxiliary hypothesis testing problem, in analogy to the technique used for the single--user case:
 \begin{equation}
 \mathcal{H}_{i_1,i_2} : \mathbf{y} = \mathbf{x}_1(u_1+i_1{\Delta}_1) + \mathbf{x}_2(u_2+i_2{\Delta}_2) + \mathbf{z}, 
 \end{equation} for $i_1 \in \{1,\cdots,M_1\} $ and $i_2 \in \{1,\cdots,M_2\} $, 
where $u_1 \in \left[0,1-(M_1-1)\Delta_1 \right)$, $u_2 \in \left[0,1-(M_2-1)\Delta_2 \right)$. Both $u_1$ and $u_2$ are known to the receiver. As in the single--user case, we will derive two types of results. The first corresponds to fixed values of $M_1$ and $M_2$ (and $\Delta_1$, $\Delta_2$), which will yield non-asymptotic results on the MSE's themselves. The second type of results refers to the asymptotic regime of large $N$, where $M_1$ and $M_2$ are allowed to grow exponentially with $N$, at arbitrary rates to be optimized, and our asymptotic results concern the asymptotic exponential rates of the two MSE's. 

\subsection{Outer bounds on the region of achievable MSE pairs \label{subsec:dual_zr_linear}}
We denote the conditional probability of error as a function of $(u_1,u_2)$ by 
\begin{equation}
P_{e}(u_1,u_2,\Delta_1,\Delta_2)=\Pr\left \{(\hat{i}_1,\hat{i}_2)\neq(i_1,i_2) | u_1,u_2\right\}
\end{equation}
where the overall probability of error is $P_{e}=\int_{u_1}du_1p(u_1)\int_{u_2}du_1p(u_2)P_{e}(u_1,u_2,\Delta_1,\Delta_2)$ with $p(.)$, $\hat{i}_1$ and $\hat{i}_2$ being the probability density function, the estimates of $i_1$ and $i_2$, respectively. As noted in Section \ref{sec:model_2}, the results in this paper are presented for the case where the sources are uniformly distributed over $[0,1)$ and the adaptation to other choices of source distributions is straightforward.
A lower bound on $P_{e}(u_1,u_2)$ will now be derived by generalizing Shannon's zero-rate lower bound for the Gaussian MAC. The overall probability of error for this channel can be decomposed into three terms as follows:
\begin{equation}
P_{e}= \Pr\left(\hat{i}_1\neq i_1, \hat{i}_2=i_2 \right)+ \Pr\left(\hat{i}_1= i_1, \hat{i}_2\neq i_2 \right) + \Pr\left(\hat{i}_1\neq i_1, \hat{i}_2\neq i_2 \right)  \label{eq:Pe-MAC1} 
\end{equation}
Here we need a two--user counterpart of $L_B(\Delta)$ (\ref{eq:LB}) which depends on two parameters, $\Delta_1,\Delta_2$ for $U_1,U_2$, respectively, that is
$$\Pr\{|\hat{U}_1(\mathbf{y})-U_1|>\Delta_1/2~\textrm{or}~|\hat{U}_2(\mathbf{y})-U_2|>\Delta_2/2\}\geq
L_B(\Delta_1,\Delta_2),$$
with the l.h.s.\ being further upper bounded using the union bound, to yield
\begin{equation} \label{eq:sum_prob}
\Pr\{|\hat{U}_1(\mathbf{y})-U_1|>\Delta_1/2\}+\Pr\{|\hat{U}_2(\mathbf{y})-U_2|>\Delta_2/2\}\geq L_B(\Delta_1,\Delta_2).
\end{equation} The lower bound $L_B(\Delta_1,\Delta_2)$ is to be specified later. 
Using considerations similar to those of the derivation in (\ref{eq:Merhav13mod}), one obtains
\begin{multline}
\int_{0}^{1-(M_1-1)\Delta_1} du_1p(u_1)\int_{0}^{1-(M_2-1)\Delta_2}du_2p(u_2)P_e(u_1,u_2,\Delta_1,\Delta_2)
\leq
\\
\frac{\left(\Pr\left\{|\hat{U}_1(\mathbf{y})-U_1|>\Delta_1/2 \right\}+\Pr\left\{|\hat{U}_2(\mathbf{y})-U_2|> \Delta_2 /2\right\}\right)}{\left\lceil 1/\Delta_1\right\rceil \left\lceil 1/\Delta_2\right\rceil}. \label{Pe_ind_MAC1}
\end{multline} The l.h.s.\ of (\ref{Pe_ind_MAC1}) is obtained by introducing the condition of $M_j\Delta_j>1$, which is equivalent to $M_j=\left\lceil 1/\Delta_j\right\rceil$, for $j=1,2$. We note that (\ref{Pe_ind_MAC1}) is valid for all $M_j$ and $\Delta_j$ such that $(M_j-1)\Delta_j<1$. A detailed derivation of (\ref{Pe_ind_MAC1}) can be found in Appendix \ref{subsec:sum_MSE_improved}.
Combining (\ref{eq:sum_prob}) and (\ref{Pe_ind_MAC1}) with (\ref{eq:Pzr}), we finally have
\begin{align}
L_B(\Delta_1,\Delta_2) \dfn &\left\lceil 1/\Delta_1\right\rceil \left\lceil 1/\Delta_2\right\rceil\left(1+\Delta_1-\left\lceil 1/\Delta_1\right\rceil\Delta_1\right)
 \left(1+\Delta_2-\left\lceil 1/\Delta_2\right\rceil \Delta_2 \right) \nonumber \\
&P_{ZR}\left(\mathcal{E}_1,\mathcal{E}_2,\left\lceil 1/\Delta_1\right\rceil,\left\lceil 1/\Delta_2\right\rceil\right). \label{Pe_ind_MAC2}
\end{align} 
\subsubsection{Shannon's zero-rate bound adapted to the MAC}
Shannon's bound is based on first upper bounding the average squared Euclidean distance between all pairs of modulated signals and this should be carried out for each of the three terms of eq. (\ref{eq:Pe-MAC1}). In the first term in (\ref{eq:Pe-MAC1}) there are $M_{1}(M_{1}-1)/2$ possible signal pairs, and so, the average squared Euclidean distance between all such pairs is upper bounded by
\begin{equation}
D_1^2 (u_1,u_2) \leq \frac{2M_{1} \mathcal{E}_1}{(M_{1}-1)} \label{eq:Dmin1}
\end{equation}
Similarly, for the second term of (\ref{eq:Pe-MAC1}),
\begin{equation}
 D_2^2 (u_1,u_2)  \leq\frac{2M_{2}\mathcal{E}_2}{(M_{2}-1)} \label{eq:Dmin2}
\end{equation} with $M_{2}(M_{2}-1)/2$ signal pairs of user 2. 
For the third term, there are $M_1M_{2}(M_1-1)(M_{2}-1)$ possible pairs that differ in both indices, so that
\begin{equation}
D_{12}^2 (u_1,u_2)  \leq\frac{2 M_1 \mathcal{E}_1}{(M_1-1)} +\frac{2 M_{2} \mathcal{E}_2}{(M_{2}-1)} \label{eq:Dmin12}
\end{equation} The reader is referred to Appendix \ref{subsec:euclid} for a detailed derivation of eqs.\ (\ref{eq:Dmin2})-(\ref{eq:Dmin12}). 
By progressively removing points at the average distance as in \cite[eq.\ (81)]{Shannon59}, we obtain the overall bound as follows.
\begin{align}
P_e (u_1,u_2,\Delta_1,\Delta_2)& \geq P_{ZR}^{\mathrm{Shannon}}(\mathcal{E}_1,\mathcal{E}_2,M_1,M_2) \notag\\
                    &= \frac{1}{M_1}\sum_{m=2}^{M_{1}}Q\left(\sqrt{\frac{m}{m-1}\frac{\mathcal{E}_1}{2\sigma^2}}\right) +  \frac{1}{M_2}\sum_{m=2}^{M_{2}}Q\left(\sqrt{\frac{m}{m-1}\frac{\mathcal{E}_2}{2\sigma^2}}\right) \notag \\
										&+ \frac{1}{M_1M_2}\sum_{m_1=2}^{M_1}\sum_{m_2=2}^{M_2}Q\left(\sqrt{\frac{m_1}{m_1-1}\frac{\mathcal{E}_1}{2\sigma^2}+\frac{m_2}{m_2-1}\frac{\mathcal{E}_2}{2\sigma^2}}\right) \label{eq:Pzr}
\end{align}
\subsubsection{An alternative zero-rate bound}
In the proof of Theorem 4 from \cite{Weng}, the authors showed that the overall error probability (\ref{eq:Pe-MAC1}) of a two-user Gaussian MAC with codebooks $\mathcal{C}_1$ and $\mathcal{C}_2$ is lower bounded by the error probability of the single-user code $\mathcal{C}_1+\mathcal{C}_2$ under an average power constraint. In our case, the resulting lower bound using an average power constraint is still valid since a peak energy/power constraint can only increase the error probability. Note that in our case the sum codebook has energy $\mathcal{E}_1+\mathcal{E}_2$ and cardinality $M_1 M_2$. The error probability of the sum codebook can then be lower bounded by (\ref{eq:Poly_single_zero}) using $\mathcal{E}_1+\mathcal{E}_2$ and $M_1 M_2$ for $\mathcal{E}$ and $M$. Including the single-user lower bounds for each user, the overall bound on the zero rate error probability is the maximum of three functions as
\begin{equation} \label{eq:Pzr_poli}
P_{ZR}^{\mathrm{P}}(\mathcal{E}_1,\mathcal{E}_2,M_1,M_2)= \max \left\{ P_\mathrm{ZR}^{\mathrm{P}}\left(\mathcal{E}_1,M_1\right), P_\mathrm{ZR}^{\mathrm{P}}\left(\mathcal{E}_2,M_2\right), P_\mathrm{ZR}^{\mathrm{P}}\left(\mathcal{E}_1+\mathcal{E}_2,M_1 M_2\right)\right\}
\end{equation} where $P_\mathrm{ZR}^{\mathrm{P}}\left(\mathcal{E},M\right)$ is given by (\ref{eq:Poly_single_zero}).
 
In the next theorem, we state the first main result for the two-user setting.
\begin{theorem} \label{theorem_2}
For arbitrary modulators $\mathbf{x}_j(U_j),\; {j=1,2}$, transmitting subject to power limitations, $\mathcal{S}_1$ and $\mathcal{S}_2$, respectively, over the two--user Gaussian MAC (\ref{eq:MAC_def}), the following inequalities hold
\begin{eqnarray} 
{\mathrm{MSE}}_1 &\geq &\max\left({\mathrm{MSE}}_{\mathrm{s},1},\max_{0<\theta\leq 1}\left(C_1(\theta)/2-\frac{{\mathrm{MSE}}_{2}}{\theta^2}\right), \max_{0<\theta\leq 1} \theta^2 \left(C_{2}(\theta)/2- {\mathrm{MSE}}_{2}\right)\right), \label{eq:twouser_MSE_imp} \\
{\mathrm{MSE}}_2 &\geq &\max\left({\mathrm{MSE}}_{\mathrm{s},2},\max_{0<\theta\leq 1}\left(C_2(\theta)/2-\frac{{\mathrm{MSE}}_{1}}{\theta^2} \right), \max_{0<\theta\leq 1} \theta^2 \left(C_{1}(\theta)/2- {\mathrm{MSE}}_{1}\right)\right), \label{eq:twouser_MSE_imp_2}
\end{eqnarray}
where $\mathrm{MSE}_{\mathrm{s},j}$ denotes the lower bound on the MSE in estimating the parameter $U_j$, $j=1,2$, in the single--user case (or equivalently, when the other parameter is known), given by (\ref{merhav-mod-lhs}), with 
\begin{eqnarray}
C_1(\theta)&=&\int_0^1d\Delta\cdot\Delta\cdot L_B(\Delta,\theta\Delta)\notag \\
C_2(\theta)&=&\int_0^1d\Delta\cdot\Delta\cdot L_B(\theta\Delta,\Delta) \notag
\end{eqnarray}
and $L_B(.,.)$ is given by (\ref{Pe_ind_MAC2}). 
\end{theorem}
\begin{proof}
Let $\theta$ be an arbitrary parameter, taking on values in $[0,1]$, and for a given $\Delta$, set $\Delta_1=\Delta$ and $\Delta_2=\theta \Delta$. Now, by integrating both sides of (\ref{eq:sum_prob}) 
w.r.t.\ $\Delta$ we have
\begin{equation} \label{eq:Sum_prob2}
\int_{0}^{1}\mbox{d}\Delta\cdot\Delta\left(\Pr\{|\hat{U}_1(\mathbf{y})-U_1|>\Delta/2\}+\Pr\{|\hat{U}_2(\mathbf{y})-U_2|>\theta\Delta/2\}\right) \geq C_1(\theta).
\end{equation} For the derivation of $C_1(\theta)$, the reader is referred to Appendix \ref{app_ctheta}. 
Now, the first term on the l.h.s.\ is upper bounded by $2\mathbf{E}[\hat{U}_1(\mathbf{y})-U_1]^2$. As for the second term, similarly, we get $$\int_0^1d\Delta\cdot\Delta\cdot\Pr\{|\hat{U}_2(\mathbf{y})-U_2|>\theta \Delta/2\} \leq\frac{2}{\theta^2}\cdot\mathbf{E}[\hat{U}_2(\mathbf{y})-U_2]^2.$$ 
Combining this with (\ref{eq:Sum_prob2}), we readily obtain
\begin{equation}
\label{2nd}
\mathrm{MSE}_1+\frac{\mathrm{MSE}_2}{\theta^2}\geq\frac{C_1(\theta)}{2}
\end{equation} or equivalently,
\begin{equation}
\mathrm{MSE}_1\geq \frac{C_1(\theta)}{2}-\frac{\mathrm{MSE}_2}{\theta^2}.
\end{equation} 
Since this inequality holds true for every $\theta\in[0,1]$, the tightest bound of this form is obtained by
maximizing the r.h.s.\ over $\theta$ in this interval, which yields
\begin{equation}
\mathrm{MSE}_1\geq \max_{0\le\theta\le 1}\left[\frac{C_1(\theta)}{2}-\frac{\mathrm{MSE}_2}{\theta^2}\right].
\end{equation}
We also observe that the single--user bound $\mathrm{MSE}_1\ge \mathrm{MSE}_{\mathrm{s},j}$ trivially holds
since it is equivalent to a ``genie-aided'' scenario, 
where user no.\ 1 is fully informed on the exact value of $U_2$.

The equivalence of (\ref{2nd}) using $C_1(\theta)$ could be given also for user 2 as
\begin{equation}
\label{3rd}
\theta^2 \mathrm{MSE}_1+\mathrm{MSE}_2\geq \theta^2\frac{C_1(\theta)}{2}.
\end{equation}
By the same token, eq.\ (\ref{3rd}) implies that
\begin{equation}
\mathrm{MSE}_2\geq \max_{0\le\theta\le 1}\theta^2\left[\frac{C_1(\theta)}{2}-\mathrm{MSE}_1\right].
\end{equation}
To obtain the remaining bounds, interchange the roles of the two users, which amounts to the use of $C_2(\theta)$. This completes the proof of Theorem 2.
\end{proof}

In Section \ref{sec:numerical} we present numerical evaluation results of (\ref{eq:twouser_MSE_imp})-(\ref{eq:twouser_MSE_imp_2}) for different values of $\theta$ and SNR.

\subsection{Upper Bounds on the MSE exponents \label{sec:nonzero}}

In this subsection, we modify the bounds presented in Theorem \ref{theorem_2} in order to obtain upper bounds of the achievable region of the MSE exponents defined as in (\ref{eq:expo_mac}). The core idea is to pass from the zero--rate bound of the previous subsection, where $M_1$ and $M_2$ were fixed (independent of $N$), to positive rate bounds, where $M_1=e^{NR_1}$ and $M_2=e^{NR_2}$, $R_1$ and $R_2$ being subjected to optimization.
Our main result, in this subsection, is asserted in the following theorem.
\begin{theorem}
For arbitrary $N$-dimensional parameter modulators $\mathbf{x}_j(U_j), \; j=1,2$ transmitting subject to power constraints given by (\ref{eq:enr}) across the two--user Gaussian MAC (\ref{eq:MAC_def}), 
the MSE exponents are bounded by 
\begin{eqnarray}\label{eq:Falpha}
\epsilon_1 &\leq & \min \left\{\epsilon_{\mathrm{s},1},\inf_{\alpha:~F(\alpha)+2\alpha\geq \epsilon_{2}}F(\alpha), \inf_{\alpha:~G(\alpha)\geq \epsilon_2}G(\alpha) +2\alpha \right\} \\
\epsilon_2 &\leq & \min \left\{\epsilon_{\mathrm{s},2},\inf_{\alpha:~G(\alpha)+2\alpha\geq \epsilon_{1}}G(\alpha), \inf_{\alpha:~F(\alpha)\geq \epsilon_1}F(\alpha) +2\alpha \right\} \label{eq:Falpha'}
\end{eqnarray} where 
\begin{eqnarray}
F(\alpha) &\dfn &  \min_R[E_{u}(R,R+\alpha)+2R]\}, \notag \\
G(\alpha) &\dfn & \min_R[E_{u}(R+\alpha,R)+2R]\} \notag
\end{eqnarray} 
and $E_{u}(R_1,R_2)$ and $\epsilon_{\mathrm{s},j}$ denote any upper bound on the reliability function of the two--user Gaussian MAC and the single-user bound on the MSE exponent in estimating the parameter $U_j$, $j=1,2$, given by Theorem \ref{t_single_nonzero}, respectively.
\end{theorem}
\begin{proof}
Substituting $\Delta=e^{-RN}$ and $\theta=e^{-\alpha N}$ into (\ref{2nd}) and
changing the integration variable on the
r.h.s. of (\ref{2nd}) 
to $R$, we obtain
\begin{equation}
\mathrm{MSE}_1+e^{2\alpha N}\mathrm{MSE}_2 \geq \frac{N}{2}\int_0^\infty dR\cdot e^{-2RN}\cdot
L_B(e^{-RN},e^{-(R+\alpha) N}) \label{eq:bir}
\end{equation}
By the Laplace integration method \cite{Laplace} the r.h.s.\ of (\ref{eq:bir}) is of the exponential order of
$\exp\{-N\min_R[E_{u}(R,R+\alpha)+2R]\}=\exp\{-NF(\alpha)\}$.
The l.h.s.\ is of the exponential order of
$\exp\{\min\{\epsilon_1,\epsilon_2-2\alpha\}\}$. Thus, we obtain
\begin{equation}
\min\{\epsilon_1,\epsilon_2-2\alpha\}\leq F(\alpha)~~~~~\forall \alpha\geq 0.
\end{equation}
In other words, for every $\alpha\geq 0$, there exists $\lambda\in[0,1]$ such that
$\lambda\epsilon_1+(1-\lambda)(\epsilon_2-2\alpha)\leq F(\alpha)$
or equivalently: 
\begin{equation} \label{eps_1_general}
\epsilon_1\leq
\inf_{\alpha\geq 0}\sup_{0\leq \lambda \leq 1}
\frac{F(\alpha)+(1-\lambda)(\epsilon_2-2\alpha)}{\lambda}=\inf_{\alpha:~F(\alpha)+2\alpha\geq \epsilon_2}F(\alpha).
\end{equation}
Substituting $\Delta=e^{-RN}$ and $\theta=e^{-\alpha N}$ into (\ref{3rd}) and
changing the integration variable on the r.h.s.  
to $R$, we get $\max\{\epsilon_1-2\alpha,\epsilon_2\}\leq G(\alpha)~,\forall \alpha\geq 0$ that yields the following bound on $\epsilon_1$ as 
\begin{equation} \label{eps_1_general_1}
\epsilon_1\leq
\inf_{\alpha\geq 0}\sup_{0\leq \lambda \leq 1}
\left(\frac{G(\alpha)+(1-\lambda)\epsilon_2}{\lambda}+2\alpha\right)=\inf_{\alpha:~G(\alpha)\geq \epsilon_2}G(\alpha)+2\alpha.
\end{equation}  The overall bound on $\epsilon_1$ is the maximum of the three bounds given by (\ref{eps_1_general}), (\ref{eps_1_general_1}) 
and the bound on the single--user MSE exponent given by (\ref{eq:expo_1}).
The bound to $\epsilon_2$ is obtained in the very same manner.
\end{proof}
For the purpose of numerical evaluation, we will study three different alternatives for $E_u(R_1,R_2)$ to be used in bounding the MSE exponents (\ref{eq:Falpha})-(\ref{eq:Falpha'}) assuming equal energy on both transmitters, i.e. $S_1=S_2=S$. Clearly, equal energy on both users will result in the same exponent $F(\alpha)$ (or $G(\alpha)$). 

\subsubsection{Divergence bound}

$E_{u}(R_1,R_2)$ is chosen as the sphere-packing bound of \cite{Nazari-thesis}, taking the auxiliary channel $W$ to be a Gaussian MAC with noise
variance $\sigma_w^2$. For inputs of powers as defined by (\ref{eq:enr}), the rate region of the auxiliary Gaussian MAC $W$ is given by
\begin{eqnarray}
R_j &\leq& \frac{1}{2}\log\left(1+\frac{\mathcal{S}}{\sigma_w^2}\right)\\
R_1+R_2 &\leq & \frac{1}{2}\log\left(1+\frac{2\mathcal{S}}{\sigma_w^2}\right),
\end{eqnarray}
which implies that for $W$ to exclude $(R_1,R_2)$ from the achievable region, 
\begin{equation}
\sigma_w^2\geq \min\left\{\frac{\mathcal{S}}{e^{2R_1}-1},
\frac{\mathcal{S}}{e^{2R_2}-1},
\frac{2\mathcal{S}}{e^{2(R_1+R_2)}-1}\right\}\dfn \sigma_0^2(R_1,R_2),
\end{equation}
and its assumed that $\sigma_0^2(R_1,R_2) > \sigma^2$. Thus, 
\begin{eqnarray}
E_{\mbox{\tiny sp}}(R_1,R_2)&=&\frac{1}{2}\left[\frac{\sigma_0^2(R_1,R_2)}{\sigma^2}-\ln\left(
\frac{\sigma_0^2(R_1,R_2)}{\sigma^2}\right)-1\right]\nonumber\\
&=&\min\{D(R_1,\mathcal{S}),D(R_2,\mathcal{S}),D(R_1+R_2,2\mathcal{S})\},
\end{eqnarray}
where the divergence function is defined using \cite[eq. (5.27)]{Nazari-thesis} as
\begin{equation}
D(R,\mathcal{S})\dfn
\frac{1}{2}\left[\frac{\mathcal{S}}{\sigma^2(e^{2R}-1)}-\ln\left(\frac{\mathcal{S}}{\sigma^2(e^{2R}-1)}\right)-1\right].
\end{equation} The derivation of the upper bound on the sphere-packing bound $E_{\mbox{\tiny sp}}(R_1,R_2)$ for the Gaussian MAC can be found in Appendix \ref{subsec:GMAC_Esp}. 
We first need to calculate
\begin{eqnarray}
F(\alpha)&=&\inf_{R>0}\left\{2R+E_{\mbox{\tiny sp}}(R,R+\alpha)\right\}\nonumber\\
&=&\inf_{R>0}\left\{2R+\frac{1}{2}\left[\frac{\sigma_0^2(R,R+\alpha)}{\sigma^2}-\ln\left(
\frac{\sigma_0^2(R,R+\alpha)}{\sigma^2}\right)-1\right]\right\}\nonumber\\
&=&\min\{F_1,F_2(\alpha),F_{12}(\alpha)\},
\end{eqnarray}
with
\begin{eqnarray} \label{eq:F_functions_div1}
F_1&=&\inf_{R\geq 0}[2R+D(R,\mathcal{S})]\\
F_2(\alpha)&=&\inf_{R\geq 0}[2R+D(R+\alpha,\mathcal{S})]  \label{eq:F_functions_div2}\\
F_{12}(\alpha)&=&\inf_{R\geq 0}[2R+D(2R+\alpha,2\mathcal{S})].  \label{eq:F_functions_div3}
\end{eqnarray}
The channel rates that minimize the three exponents $F_1$, $F_2(\alpha)$ and $F_{12}(\alpha)$ given by (\ref{eq:F_functions_div1})-(\ref{eq:F_functions_div3}) are denoted respectively by $R_1^*$, $R_2^*$ and $R_{12}^*$ that are derived and given in detail in Appendix \ref{subsec:expo_deriv}.
Using these rate functions we can reformulate the minimum functions $F_1^*$, $F_2^*(\alpha)$ and $F_{12}^*(\alpha)$ as functions of  $R_1^*$, $R_2^*$ and $R_{12}^*$, respectively.
Considering the constraint in (\ref{eq:Falpha}), we choose the $\alpha$ satisfying
\begin{equation} \label{eq:esp2_divergence}
\epsilon_2 \leq \min\{F_1^*,F_2^*(\alpha),F_{12}^*(\alpha)\}+2\alpha.
\end{equation}
The constraint $\epsilon_2 \leq F_1^*+2\alpha$ yields
\begin{equation}
\alpha \leq \frac{F_1^*-\epsilon_2}{2}\dfn\alpha_1(\epsilon_2).
\end{equation}
The constraint $\epsilon_2\leq F_2^*(\alpha)+2\alpha$ gives no requirement
concerning $\alpha$, it is simply the single-user bound for user 2.
For the two-user component $\epsilon_2\leq F_{12}^*(\alpha)+2\alpha$ we get
\begin{equation}
\alpha\leq \frac{1}{2}(F_{12}^*(\alpha)-\epsilon_2)\dfn \alpha_2(\epsilon_2).
\end{equation}
Thus, the constraint becomes
\begin{equation}
\alpha\leq\alpha^*(\epsilon_2)\dfn\max\{\alpha_1(\epsilon_2),\alpha_2(\epsilon_2)\}
\end{equation}
resulting in the overall bound
\begin{eqnarray}  \label{eq:esp1_divergence}
\epsilon_1&\leq &F[\alpha^*(\epsilon_2)]\nonumber\\
&=&\min\{F_1,F_2[\alpha^*(\epsilon_2)], F_{12}[\alpha^*(\epsilon_2)]\}.
\end{eqnarray}
The roles of the users should be interchanged to obtain the upper
bound for $\epsilon_2$ as a function of $\epsilon_1$. The overall upper bound on the achievable region of the MSE exponents is the
intersection of the two.
Note that the upper bound on the MSE exponent in a point-to-point channel that is derived from (\ref{eq:single_nonzero}) in the previous part is equivalent to (\ref{eq:F_functions_div1}). 

\subsubsection{Shannon's sphere-packing bound \label{subsec:Sh_MAC}}

As a second alternative to the divergence bound by Nazari, we adopt Shannon's sphere-packing bound studied Section \ref{subsubsec:sphere_packing} to the two-user setting. 
Before defining the exponents using (\ref{eq:Shannon_rel}), we remind the reader concerning the error exponent region for a MAC, introduced in \cite[Theorem 4]{Weng}. The authors of \cite{Weng} show that for the Gaussian MAC with equal signal powers, denoted by $S$, an outer bound on the error exponent region is dictated by three inequalities. The first two error exponents $E_j, \; j=1,2$ are bounded from above by $E_{su}(R_j,\mathcal{S}/\sigma^2)$ and correspond to the two single-user error events, and the third exponent $E_{su}(R_1+R_2, 2S/\sigma^2)$ corresponds to the joint error event. In all inequalities, $E_{su}(R)$ represents any upper bound on the reliability function of the single-user AWGN channel. Let us denote the three exponents which make use of (\ref{eq:Shannon_rel}) in the minimization by $F_{1,\mbox{\tiny Sh}}$, $F_{2,\mbox{\tiny Sh}}(\alpha)$ and for the two-user component by $F_{12,\mbox{\tiny Sh}}(\alpha)$.

Using the results of \cite{Weng}, the single-user components are functions of the minimum rate given as (\ref{eq:Rmin_sh}). 
\begin{align} \label{eq:F1_Sh}
F_{1,\mbox{\tiny Sh}}^*&= 2R_{\min}+E_{\mbox{\tiny sp}}(\psi(R_{\min}), A) \\
F_{2,\mbox{\tiny Sh}}^*(\alpha) &=F^*_{1,\mbox{\tiny Sh}}-2\alpha \label{eq:F2_Sh}
\end{align}
$F_{12,\mbox{\tiny Sh}}(\alpha)$ has to be optimized numerically since it does not lend itself to closed form analysis. Using (\ref{eq:Shannon_rel}) the third exponent as the two-user component is
\begin{equation}\label{eq:F12_Sh}
F_{12,\mbox{\tiny Sh}}(\alpha)= \min_{R'\geq\frac{\alpha}{2}}[2R'-\alpha+ E_{\mbox{\tiny sp}}(\psi(2R'), 2A)]
\end{equation} where $R'=R+\alpha/2$ and $A=\mathcal{S}/\sigma^2$. Similarly the two--user component $F_{12,\mbox{\tiny Sh}}(\alpha)$ with the minimum rate is denoted by $F_{12,\mbox{\tiny Sh}}^*(\alpha)$. 
The derivation of the bounds on the error exponents follow through in the same way as shown in the previous case that makes use of the divergence bound by simply replacing the three exponents in (\ref{eq:esp1_divergence}) by $F_{1,\mbox{\tiny Sh}}^*$, $F_{2,\mbox{\tiny Sh}}^*$ and $F_{12,\mbox{\tiny Sh}}^*$.


\subsubsection{The upper bound by Ashikhmin et. al.  \label{subsec:Ash_MAC}}

As for the third alternative for $E_u(R_1,R_2)$, we have a more recent result by Ashikhmin {\it et al.} \cite[Theorem 1]{Ashikhmin}, which is a tighter bound on the reliability function $E(R,A)$ with SNR $A$, and we denote it by $E_{\mbox{\tiny ABL}}(R,A)$. Note that $E_{\mbox{\tiny ABL}}(R,A)$ coincides with (\ref{eq:Shannon_rel}) above a certain rate. It is, in fact, a convex combination of (\ref{eq:Shannon_rel}) with a tighter low-rate bound which coincides with the zero-rate exponent unlike (\ref{eq:Shannon_rel}). 
We were not able to characterize the MSE exponents analytically for the Ashikhmin {\em et al.} upper bound on the reliability function. Similar to the Shannon's sphere--packing bound, we denote the three error exponents by $F_{1,\mbox{\tiny ABL}}$, $F_{2,\mbox{\tiny ABL}}(\alpha)$ and $F_{12,\mbox{\tiny ABL}}(\alpha)$, which are evaluated as
\begin{align}
&F_{1,\mbox{\tiny ABL}}= \min_{R\geq 0} {2R+E_{\mbox{\tiny ABL}}(R, A)} \nonumber \\ 
&F_{2,\mbox{\tiny ABL}}(\alpha) =F_{1,\mbox{\tiny ABL}}-2\alpha \label{eq:F2_Ash}  \\
&F_{12,\mbox{\tiny ABL}}(\alpha) = \min_{R\geq 0}[2R+ E_{\mbox{\tiny ABL}}(2R+\alpha, 2A)] \nonumber 
\end{align} where $R'=R+\alpha/2$. 
The optimal values are replaced in (\ref{eps_1_general}) to determine the MSE exponents.
It should be mentioned that the MSE exponent region in this case may coincide for some choice of SNR with the region based on (\ref{eq:Shannon_rel}) since the two error exponents coincide for some rates.  
In Section \ref{sec:numerical} the three bounds on the MSE exponents in a two-user MAC are numerically evaluated and their performances are compared as a function of various values of SNR.

\section{Numerical Results \label{sec:numerical}}
In Figure \ref{fig:lowerbounds}, we first present a numerical evaluation of the bounds for the single-user problem that was treated in Section \ref{sec:single_user} with several bounds proposed for the same problem from the literature alongside one achievable scheme.
Following the order of the curves in the legend, $M$-ary Scalar Quantization and $M$-ary Simplex refers to the exact MSE of a uniform scalar quantizer with $\log_2M$ bits that is mapped to a regular $M$-ary simplex. Note that this combination has an exponential behaviour as $O(e^{-\mathcal{E}/6})$ which is higher than that of all the lower bounds.  We also show the rate-distortion lower bound from Goblick \cite{goblick} $D=\frac{1}{2\pi e}e^{-\mathcal{E}}$ for the sake of comparison. The four remaining lower bounds make use of the results from \cite{Merhav12b} and the work reported here. The two new lower bounds correspond to (\ref{eq:Poly_single_zero}) and (\ref{eq:Sh_single_zero}) combined with (\ref{merhav-mod-lhs}).  The previously best lower bound corresponds to (\ref{eq:Sh_single_zero}) combined with the lower bound through the use of \cite[eq. 13]{Merhav12b}. We also show a conjectured bound which results from the combination of (\ref{merhav-mod-lhs}) with the exact error-probability of a regular $M$-ary simplex. The validity of this bound depends on the validity of the Weak Simplex Conjecture.  It is interesting to note that the bound obtained through the use of (\ref{eq:Poly_single_zero}) with (\ref{merhav-mod-lhs}) comes very close to the conjectured bound even for moderate signal energies.

In Figure \ref{fig:numeric3}, 
we present numerical evaluation of (\ref{eq:twouser_MSE_imp}) for different values of $\theta$. Note that signal-to-noise ratio (SNR) which is chosen equal for both transmitters as $\mathcal{E}/\sigma^2$. 
The \textit{wall} and \textit{floor}, the vertical and horizontal parts of the black curve to the axes, correspond to $\mathrm{MSE}_{\mathrm{single},j}$. The red and blue curves represent all possible bounds for ${\theta \in[0,1)}$. The convex hulls are depicted in solid and dotted black curves using the two-user adaption of the classical Shannon's zero-rate bound given by (\ref{eq:Pzr}) and the lower bound given by (\ref{eq:Pzr_poli}).

In Figure \ref{fig:numeric2}, the three bounds on the MSE exponents are numerically evaluated for different values of SNR, which is chosen equal for both transmitters. Clearly, the divergence bound is the weakest one for all values of SNR, whereas the outer bound evaluated using the reliability function bound by Ashikhmin {\em et al.}, labeled as ABL in the legend, is the tightest. It seems to coincide with the bound using (\ref{eq:Shannon_rel}) for high SNR levels in the portion not dominated by the single-user error-event. It is worth mentioning the difference between the performance of the divergence bound and reliability function is most significant for low SNR levels.
\begin{figure}[htp] 
\centering
\includegraphics[width=0.7\linewidth]{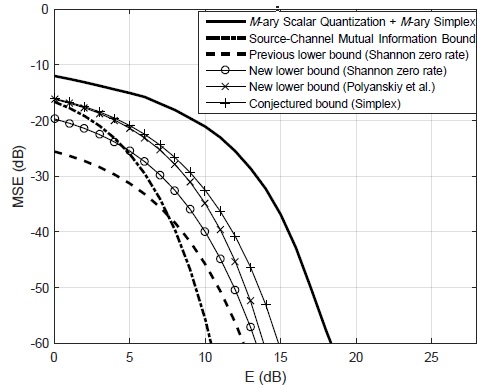}
\caption{Comparison of lower bounds on the MSE for a point-to-point channel} \label{fig:lowerbounds}
\end{figure}
\begin{figure}[H]
 \centering
    \begin{subfigure}[htp]{0.48\textwidth}
        \includegraphics[width=1.0\linewidth]{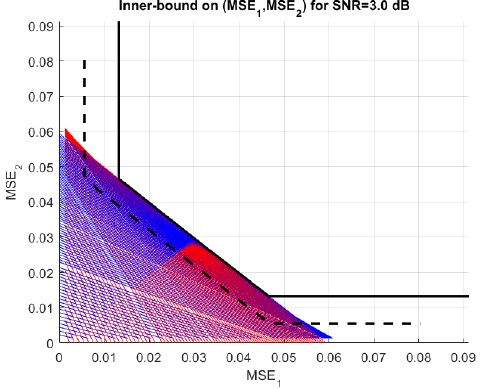}
         \label{fig:numeric3a}
    \end{subfigure}
    \begin{subfigure}[htp]{0.48\textwidth}
       \includegraphics[width=1.0\linewidth]{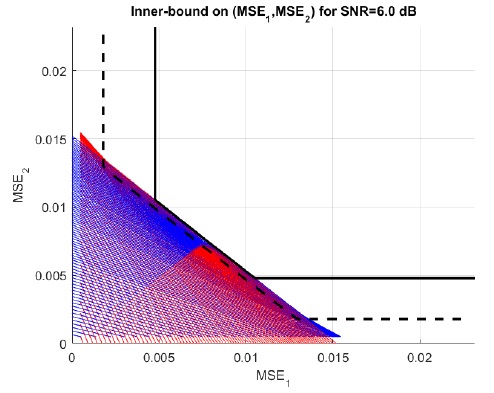}
        \label{fig:numeric3b}
    \end{subfigure}
  \begin{subfigure}[htp]{0.48\textwidth}
       \includegraphics[width=1.0\linewidth]{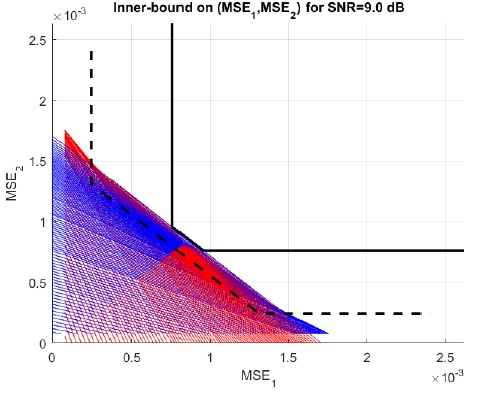}
        \label{fig:numeric3c}
    \end{subfigure}
\caption{Numerical evaluation of (\ref{eq:twouser_MSE_imp}) for different values of SNR and all possible values of $\theta$ where the dotted and solid boundaries  represent the bounds using (\ref{eq:Pzr}) and (\ref{eq:Pzr_poli}), respectively.
}   \label{fig:numeric3}
\end{figure}

\begin{figure} [htp]
\centering
\includegraphics[width=.7\linewidth]{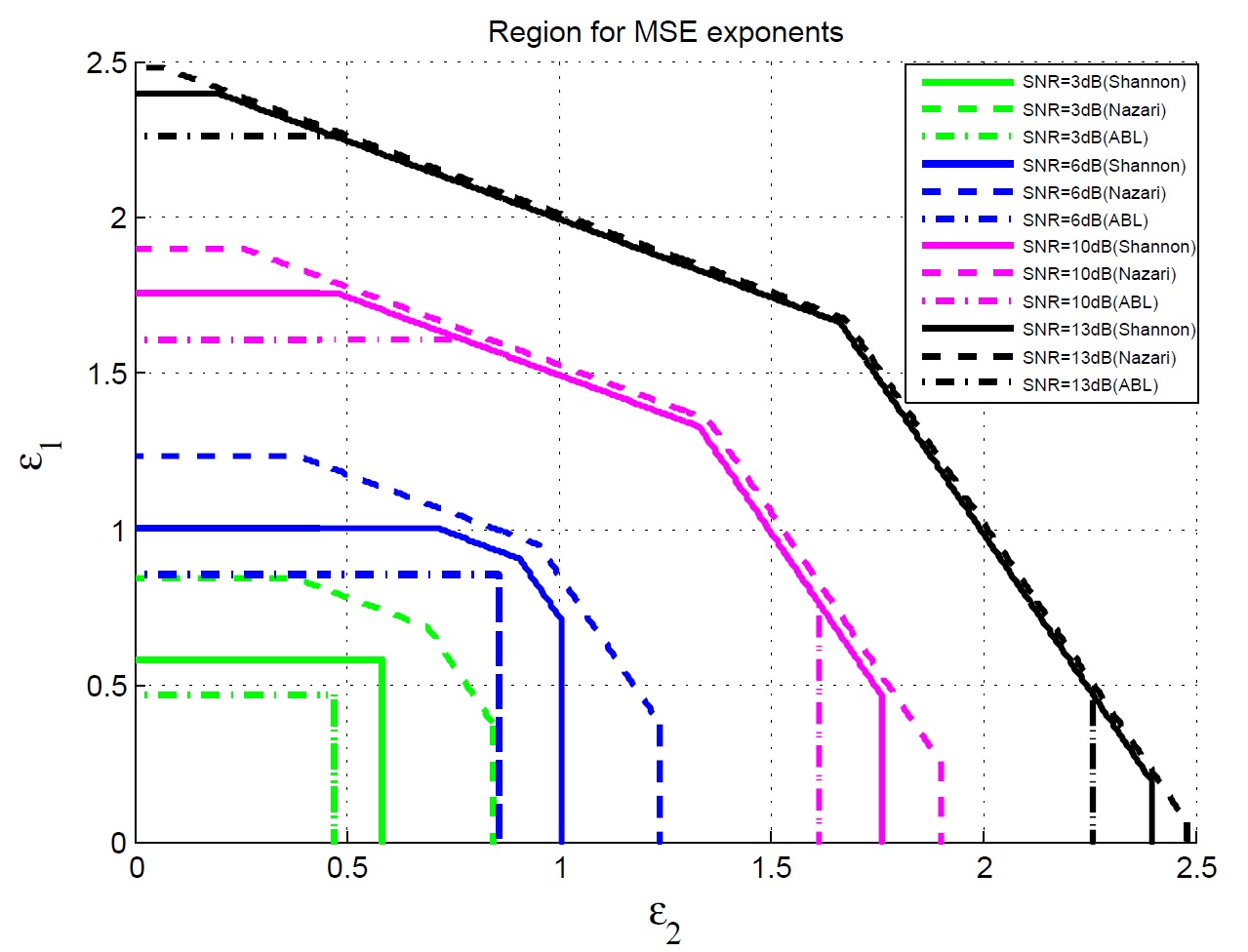}
\caption{Numerical evaluation of the upper bounds on the error exponents for different values of SNR.}
\label{fig:numeric2}
\end{figure}

\section{Conclusion \label{sec:conc}}

New lower bounds on any linear combination of the MSE's are derived for two-user separate modulation and joint estimation of parameter on a discrete-time Gaussian MAC without bandwidth constraints. To this end, we used zero-rate lower bounds on the error probability of Gaussian channels by Shannon and Polyanskiy \textit{et al.}. Numerical results showed that, the multi-user adaptation of the zero-rate lower bound by Polyanskiy \textit{et al.} provides a tighter overall lower bound on the MSE pairs than the classical Shannon bound. Additionally, we introduced upper bounds on the MSE exponents that could make use of any bound on the error exponent of a single-user AWGN channel. The obtained results are numerically evaluated for three different bounds on the reliability function of the Gaussian channel. It is shown that applying the reliability function by Ashikhmin \textit{et al.} \cite{Ashikhmin} to the MAC provides a significantly tighter characterization than Shannon's sphere-packing bound \cite{Shannon59} and the divergence bound \cite{Nazari-thesis}.
\section{Appendix \label{sec:App}}
\subsection{The derivation of the new zero-rate lower bound \label{subsec:App_poli}}

\cite[Theorem 41]{Polyanskiy} provides a lower bound on the average error probability for the AWGN channel as a function of the statistics of two random variables $H_N$ and $G_N$. Specifically, $H_N$ is defined as \cite[eq. 205]{Polyanskiy}
\begin{equation}
H_N = C+\frac{\log_2 e}{2}\frac{(2^{2C/N}-1)}{2^{2C/N}}\sum_{i=1}^{n}\left(1-Z_i^2+\frac{2\sigma}{\sqrt{\mathcal{S}}}Z_i\right),
\end{equation}
where $C=\frac{N}{2}\log_2\left(1+\frac{\mathcal{S}}{\sigma^2}\right)$ and $Z_i$ are all i.i.d. $\mathcal{N}(0,1)$. In order to simplify this for the finite-energy case, consider the random variables $Q_0=\frac{1}{\sqrt{N}}\sum_{i=1}^{N}Z_i$ so that $Q_0\sim\mathcal{N}(0,1)$ and $Q_{1,N}=\frac{1}{N}\sum_{i=1}^{N}Z_i^2$, so that $\mathrm{Var}(Q_{1,N})=\frac{2}{N}$. The first condition for the Polyanskiy {\em et al.} converse is that
\begin{equation}
\Pr\left(H_N\geq\gamma_n\right) = 1-\epsilon(\mathcal{E},M,N) \label{eq:Hconstraint}
\end{equation}
where $\epsilon(\mathcal{E},M,N)$ is the average probability of error. Expressing the right-hand tail of the c.d.f. of $H_N$ in terms of $Q_0$ and $Q_1$ yields
\begin{equation}
\Pr\left(H_N\geq\gamma_N\right) = \Pr\left(C+\frac{N(2^{2C/N}-1)\log_2 e}{2^{2C/N +1}}\left(1-Q_{1,N}\right)+\frac{N (2^{2C/N}-1)\log_2 e}{2^{2C/N}}Q_0 \geq \gamma_n\right) \label{eq:Heq}
\end{equation} 
and rearranging (\ref{eq:Heq}) in terms of $Q_0$ provides
\begin{equation}
\Pr\left(H_N\geq\gamma_N\right) = \Pr\left(Q_0\geq\frac{(\gamma_N-C)}{\log_2 e}\frac{2^{2C/N}}{\sqrt{\mathcal{E}}/\sigma}+\frac{\sqrt{\mathcal{E}}}{2\sigma}\left(1-Q_{1,N}\right)\right) \label{eq:poly1}
\end{equation}  
Now, $1-Q_{1,N}$ converges to 0 with $N$,  so we have the following  bound on (\ref{eq:poly1}) which is tight for large $N$ and some $\mu_N>0$
\begin{align*} \label{eq:Heq2}
\Pr\left(H_N\geq\gamma_N\right) &\leq \Pr(Q_{1,N}\leq 1+\mu_N) \Pr\left(Q_0\geq\frac{(\gamma_N-C)}{\log_2 e}\frac{2^{2C/N}}{\sqrt{\mathcal{E}}/\sigma}-\mu_N\frac{\sqrt{\mathcal{E}}}{2\sigma}\right) + \Pr(Q_{1,N}>1+\mu_N)\notag \\
\end{align*}
\begin{align}
&\leq \Pr\left(Q_0\geq\frac{(\gamma_N-C)}{\log_2 e}\frac{2^{2C/N}}{\sqrt{\mathcal{E}}/\sigma}-\mu_N\frac{\sqrt{\mathcal{E}}}{2\sigma}\right) + \delta_N \notag \\
&=Q\left(\frac{(\gamma_N-C)}{\log_2 e}\frac{2^{2C/N}}{\sqrt{\mathcal{E}}/\sigma}-\mu_N\frac{\sqrt{\mathcal{E}}}{2\sigma}\right)+ \delta_N 
\end{align}
where $\delta_N=\Pr(Q_{1,N}>1+\mu_N)=1-\frac{1}{\Gamma\left(\frac{N}{2}\right)}\gamma\left(\frac{N}{2},\frac{N(1+\mu_N)}{2}\right)\leq\left(1+\mu_N\right)e^{-\frac{N\mu_N}{2}}$ \cite[p.1325,Lemma 1]{Massart}.
Combining (\ref{eq:Heq2}) with (\ref{eq:Hconstraint}) yields
\begin{equation} \label{ineq:poli_eps}
\frac{(\gamma_N-C)}{\log_2 e}\frac{2^{2C/N}}{\sqrt{\mathcal{E}}/\sigma}-\mu_N\frac{\sqrt{\mathcal{E}}}{2\sigma} \leq Q^{-1}\left(1-\epsilon(\mathcal{E},M,N)-\delta_N\right)
\end{equation}

Turning now to the $G_N$, from \cite[eq. 204]{Polyanskiy} we have
\begin{align}
G_N &= C-\frac{(2^{2C/N}-1)\log_2 e}{2}\sum_{i=1}^{N}\left(1+Z_i^2-2\sqrt{1+\frac{\sigma^2}{\mathcal{S}}}Z_i\right) \notag \\
&=C-\frac{\mathcal{E}\log_2 e}{2\sigma^2}\left(1+Q_{1,N}\right)+\frac{\log_2 e}{\sigma}2^{C/N} \sqrt{\mathcal{E}} Q_0
\end{align}
Rearranging $\Pr\left(G_N\geq\gamma_N\right)$ in terms of $Q_0$ yields
\begin{align}
\Pr\left(G_N\geq\gamma_N\right) &= \Pr\left(Q_0\geq\frac{(\gamma_N-C)}{\log_2 e\sqrt{(\mathcal{E}/\sigma)}2^{C/N}}+ \frac{\sqrt{\mathcal{E}}}{2^{C/N}}\frac{1+Q_{1,N}}{2\sigma}\right)\notag \\
&\geq(1-\delta_N)\Pr\left(\left. Q_0\geq\frac{(\gamma_N-C)}{\log_2 e\sqrt{(\mathcal{E}/\sigma)}2^{C/N}}+ \frac{\sqrt{\mathcal{E}}}{2^{C/N}}\frac{1+Q_{1,N}}{2\sigma}\right \rvert Q_{1,n}\leq 1+\mu_N\right)\notag \\
&\overset{(a)}{\geq}(1-\delta_N)\Pr\left(Q_0\geq \frac{Q^{-1}\left(1-\epsilon(\mathcal{E},M,N)-\delta_N\right)}{2^{3C/N}} + \mu_N\frac{\sqrt{\mathcal{E}}}{\sigma 2^{3C/N+1}} + \frac{\sqrt{\mathcal{E}}}{2^{C/N}}\left(\sigma^{-1}+\frac{\mu_N}{2\sigma}\right)\right)\notag \\
&=(1-\delta_N)Q\left(\frac{Q^{-1}\left(1-\epsilon(\mathcal{E},M,N)-\delta_N\right)}{2^{3C/N}} + \mu_N\frac{\sqrt{\mathcal{E}}}{\sigma 2^{3C/N+1}} + \frac{\sqrt{\mathcal{E}}}{2^{C/N}}\left(\sigma^{-1}+\frac{\mu_N}{2\sigma}\right)\right)
\end{align} where step (a) is obtained using (\ref{ineq:poli_eps}). 
Polyanskiy's bound in \cite[eq.208]{Polyanskiy} on the signal-set cardinality becomes
\begin{equation}
M\leq\frac{1}{\Pr(G_N\geq\gamma_N)}\leq \left[(1-\delta_N)Q\left(\frac{Q^{-1}\left(1-\epsilon(\mathcal{E},M,N)-\delta_N\right)}{2^{3C/N}} + \mu_N\frac{\sqrt{\mathcal{E}}}{\sigma 2^{3C/N+1}} + \frac{\sqrt{\mathcal{E}}}{2^{C/N}}\left(\sigma^{-1}+\frac{\mu_N}{2\sigma}\right)\right)\right]^{-1}
\end{equation}
which when rearranged for the error probability becomes
\begin{equation}
\epsilon(\mathcal{E},M,N)\geq Q\left(\frac{\sqrt{\mathcal{E}}}{\sigma}\left(\left(1+\frac{\mathcal{E}}{N\sigma^2}\right)\left(1+\frac{\mu_N}{2}\right)+\frac{\mu_N}{2}\right)-\left(1+\frac{\mathcal{E}}{N\sigma^2}\right)^{3/2}Q^{-1}\left(\frac{1}{M(1-\delta_N)}\right)\right)-\delta_N
\end{equation}
Now, $\lim_{N\rightarrow\infty}\delta_N=0$, so the limiting expression becomes
\begin{equation} 
\lim_{N\rightarrow\infty}\epsilon(\mathcal{E},M,N)\geq Q\left(\frac{\sqrt{\mathcal{E}}}{\sigma}(1+\mu)-Q^{-1}\left(\frac{1}{M}\right)\right)
\end{equation}
for any arbitrarily small $\mu>0$. The obtained bound is given by (\ref{eq:Poly_single_zero}) in Section \ref{subsec:poli}.

\subsection{ The average squared Euclidean distance derivation for a two-user MAC \label{subsec:euclid}}
The average squared Euclidean distance for the pairs represented by the first term in (\ref{eq:Pe-MAC1}) is given by
\begin{align}
 D_2^2 (u_1,u_2) &= \frac{1}{M_{2}(M_{2}-1)}\sum_{i'_1=1}^{M_{2}}\sum_{i'_2=1}^{M_{2}} \sum_{n=1}^{N}\left|x_{2,i'_1,n}-x_{2,i'_2,n}\right|^2 \notag\\
                             &= \frac{2}{M_{2}(M_{2}-1)}\left[M_{2}\sum_{i'=1}^{M_{2}}\left\|\mathbf{x}_{2,i'}\right\|^2 - \sum_{n=1}^{N}\left|\sum_{i'}x_{2,i',n}\right|^2\right]\notag\\
              & \leq\frac{2}{(M_{2}-1)}\sum_{i'=1}^{M_{2}}\left\|\mathbf{x}_{2,i'}\right\|^2 \notag\\
                             & \leq\frac{2M_{2}}{(M_{2}-1)}\mathcal{E}_2 
\end{align} Note that the derivation given above applies to $D_1^2 (u_1,u_2)$ with $M_1$ as well. 
For the third term we have
\begin{align}
&D_{12}^2 (u_1,u_2) = \frac{1}{M_1M_{2}(M_1-1)(M_{2}-1)}\sum_{i_1=1}^{M_1}\sum_{i'_1=1}^{M_{2}}\sum_{i_2\neq i_1}\sum_{i'_2\neq i'_1} \sum_{n=1}^{N}\left|(x_{1,i_1,n}-x_{1,i_2,n})+(x_{2,i'_1,n}-x_{2,i'_2,n})\right|^2 \notag\\
&= \frac{1}{M_1 (M_{1}-1)}\sum_{i_1=1}^{M_1}\sum_{i_2\neq i_1}\sum_{n=1}^{N} |x_{1,i_1,n}-x_{1,i_2,n}|^2+
\frac{1}{M_2(M_{2}-1)}\sum_{i_1'=1}\sum_{i_2'\neq i_1'}\sum_{n=1}^{N} |x_{2,i'_1,n}-x_{2,i'_2,n}|^2+\notag\\
 &\;\;\;\;\;\frac{2}{M_1M_{2}(M_1-1)(M_{2}-1)}\sum_{i_1=1}^{M_1}\sum_{i'_1=1}^{M_{2}}\sum_{i_2\neq  i_1}\sum_{i'_2\neq i'_1} \sum_{n=1}^{N}\mathrm{Re}\left((x_{1,i_1,n}-x_{1,i_2,n})(x_{2,i'_1,n}-x_{2,i'_2,n})^*\right) \notag \\
&= \frac{1}{M_1 (M_{1}-1)}\sum_{i_1=1}^{M_1}\sum_{i_2=1}^{M_1}\sum_{n=1}^{N} |x_{1,i_1,n}-x_{1,i_2,n}|^2+
\frac{1}{M_2(M_{2}-1)}\sum_{i_1'=1}^{M_2}\sum_{i_2'=1}^{M_2}\sum_{n=1}^{N} |x_{2,i'_1,n}-x_{2,i'_2,n}|^2+\notag\\
 &\;\;\;\;\;\frac{2}{M_1M_{2}(M_1-1)(M_{2}-1)}\sum_{i_1=1}^{M_1}\sum_{i'_1=1}^{M_{2}}\sum_{i_2=1}^{M_1}\sum_{i'_2=1}^{M_2} \sum_{n=1}^{N}\mathrm{Re}\left((x_{1,i_1,n}-x_{1,i_2,n})(x_{2,i'_1,n}-x_{2,i'_2,n})^*\right) \notag \\
&=\frac{1}{M_1 (M_{1}-1)}\sum_{i_1=1}^{M_1}\sum_{i_2=1}^{M_1}\sum_{n=1}^{N} |x_{1,i_1,n}-x_{1,i_2,n}|^2+
\frac{1}{M_2(M_{2}-1)}\sum_{i_1'=1}^{M_2}\sum_{i_2'=1}^{M_2}\sum_{n=1}^{N} |x_{2,i'_1,n}-x_{2,i'_2,n}|^2+\notag\\
&\;\;\;\;\; \frac{2}{M_1M_{2}(M_1-1)(M_{2}-1)} \mathrm{Re} \left(\underbrace{\sum_{i_1=1}^{M_1}\sum_{i_2=1}^{M_{1}}\sum_{n=1}^{N}(x_{1,i_1,n}-x_{1,i_2,n})}_{0} \underbrace{\sum_{i'_1=1}^{M_{2}}\sum_{i'_2=1}^{M_2} \sum_{n=1}^{N}(x_{2,i'_1,n}-x_{2,i'_2,n})^* }_{0}\right) \notag \\
& =\frac{2}{M_1(M_1-1)}\left[M_1\sum_{i=1}^{M_1}||\mathbf{x}_{1,i}||^2 - \sum_{n=1}^{N}\left|\sum_{i=1}^{M_{1}}x_{1,i,n}\right|^2\right] + \frac{2}{M_{2}(M_{2}-1)}\left[M_{2}\sum_{i'=1}^{M_{2}}||\mathbf{x}_{2,i'}||^2 - \sum_{n=1}^{N}\left|\sum_{i'=1}^{M_2}x_{2,i',n}\right|^2\right]  \notag\\
 & \leq\frac{2 M_1}{(M_1-1)}\mathcal{E}_1 +\frac{2 M_{2}}{(M_{2}-1)}\mathcal{E}_2 
\end{align}

\subsection{Bounding the error probability in a MAC  \label{subsec:sum_MSE_improved} }

Here we will apply the modification applied to the single-user derivation that resulted in the improved lower bound (\ref{merhav-mod-lhs}) to the two-user MAC. The upper bound on the overall error probability given by (\ref{Pe_ind_MAC1}) is derived as follows
\begin{align} \label{eq:pe_mod}
&P_e=\int_{0}^{1-(M_1-1)\Delta_1} du_1p(u_1)\int_{0}^{1-(M_2-1)\Delta_2}du_2p(u_2)P_e(u_1,u_2) \notag\\
& \leq
\frac{1}{M_{1}M_{2}}\sum_{i=1}^{M_1}\sum_{i'=1}^{M_{2}}\int_{0}^{1-(M_1-1)\Delta_1} du_1 \int_{0}^{1-(M_2-1)\Delta_2}du_2 \Pr \left \{ |U_1-\hat{U}_1(\mathbf{y})| > \Delta_1/2 | U_1=u_1+i\Delta_1,U_2=u_2+i'\Delta_2 \right \} \notag\\
&+ \frac{1}{M_{1}M_{2}}\sum_{i=1}^{M_1}\sum_{i'=1}^{M_{2}}\int_{0}^{1-(M_1-1)\Delta_1} du_1\int_{0}^{1-(M_2-1)\Delta_2}du_2 \Pr\left\{|U_2-\hat{U}_2(\mathbf{y})| >\Delta_2/2| U_1=u_1+i\Delta_1, U_2=u_2+i'\Delta_2 \right \} \notag \\
&= \frac{1}{M_{1}M_{2}}\sum_{i=1}^{M_1}\sum_{i'=1}^{M_{2}}\int_{i \Delta_1}^{1-(M_1-1)\Delta_1+i \Delta_1} du_1\int_{i' \Delta_2}^{1-(M_2-1)\Delta_2+i'\Delta_2}du_2 \Pr \left \{ |U_1-\hat{U}_1(\mathbf{y})| > \Delta_1/2| U_1=u_1, U_2=u_2 \right \} \notag \\
&+ \frac{1}{M_{1}M_{2}}\sum_{i=1}^{M_1}\sum_{i'=1}^{M_{2}}\int_{i \Delta_1}^{1-(M_1-1)\Delta_1+i \Delta_1} du_1\int_{i' \Delta_2}^{1-(M_2-1)\Delta_2+i'\Delta_2}du_2
\Pr\left\{|U_2-\hat{U}_2(\mathbf{y})| >\Delta_2/2 | U_1=u_1, U_2=u_2 \right \} \notag \\
&= \frac{1}{M_{1}M_{2}}\sum_{i=1}^{M_1}\sum_{i'=1}^{M_{2}}\Pr \left \{ |U_1-\hat{U}_1(\mathbf{y})| > \frac{\Delta_1}{2} | i\Delta_1 \leq U_1\leq 1-(M_1-1)\Delta_1+i\Delta_1, i'\Delta_2 \leq U_2\leq 1-(M_2-1)\Delta_2+i'\Delta_2\right\} \notag \\
&+\frac{1}{M_{1}M_{2}}\sum_{i=1}^{M_1}\sum_{i'=1}^{M_{2}} \Pr\left\{|U_2-\hat{U}_2(\mathbf{y})| >\frac{\Delta_2}{2} | i\Delta_1 \leq U_1\leq 1-(M_1-1)\Delta_1+i\Delta_1,  i'\Delta_2 \leq U_2\leq 1-(M_2-1)\Delta_2+i'\Delta_2 \right\}
\end{align} 
We set the following relationships as $M_1=\left\lceil 1/\Delta_1\right\rceil$ and $M_2=\left\lceil 1/\Delta_2\right\rceil$ so that the lower bound $L_B(\Delta_1,\Delta_2)$ becomes
\begin{align}\label{eq:LBdelta1delta2}
&\frac{1}{\left\lceil 1/\Delta_1\right\rceil \left\lceil 1/\Delta_2\right\rceil} \sum_{i=0}^{\left\lceil 1/\Delta_1\right\rceil-1}\sum_{i'=0}^{\left\lceil 1/\Delta_2\right\rceil-1}  \nonumber \\
& \left[\Pr \left \{ |U_1-\hat{U}_1(\mathbf{y})| > \frac{\Delta_1}{2} | i\Delta_1 \leq U_1\leq 1-(M_1-1)\Delta_1+i\Delta_1,  i'\Delta_2 \leq U_2\leq 1-(M_2-1)\Delta_2+i'\Delta_2\right\}\right. \notag \\
&\left.+\Pr \left \{ |U_2-\hat{U}_2(\mathbf{y})| > \frac{\Delta_2}{2} | i\Delta_1 \leq U_1\leq 1-(M_1-1)\Delta_1+i\Delta_1,  i'\Delta_2 \leq U_2\leq 1-(M_2-1)\Delta_2+i'\Delta_2\right\}\right] \notag \\
&=\frac{1}{\left\lceil 1/\Delta_1\right\rceil \left\lceil 1/\Delta_2\right\rceil} \left [\Pr\left\{|\hat{U}_1(\mathbf{y})-U_1|>\Delta_1/2\right\}+\Pr\left\{|\hat{U}_2(\mathbf{y})-U_2|>\Delta_2/2\right\}\right] \notag \\
&\geq \left(1+\Delta_1-\left\lceil\frac{1}{\Delta_1}\right\rceil\Delta_1\right)\left(1+\Delta_2-\left\lceil\frac{1}{\Delta_2}\right\rceil\Delta_2\right)
P_{ZR}\left(\mathcal{E}_1,\mathcal{E}_2,\left\lceil\frac{1}{\Delta_1}\right\rceil,\left\lceil\frac{1}{\Delta_2}\right\rceil\right) 
\end{align} 

\subsection{Derivation of $C_1(\theta)$ \label{app_ctheta}}
As in Theorem \ref{theorem_2}, setting $\Delta_2=\theta \Delta$ and $\Delta_1=\Delta$ in (\ref{eq:LBdelta1delta2}) yields $C_1(\theta)$ as follows.
\begin{align}
&\int_{0}^{1}  d \Delta  \Delta \left(\left\lceil\frac{1}{\Delta}\right\rceil+\Delta \left\lceil\frac{1}{\Delta}\right\rceil-\left\lceil\frac{1}{\Delta}\right\rceil^2 \Delta \right)\left(\left\lceil\frac{1}{\theta \Delta}\right\rceil+\theta \Delta \left\lceil\frac{1}{\theta \Delta}\right\rceil-\left\lceil\frac{1}{\theta \Delta}\right\rceil^2 \theta \Delta \right)
P_{ZR}\left(\mathcal{E}_1,\mathcal{E}_2,\left\lceil\frac{1}{\Delta}\right\rceil,\left\lceil\frac{1}{\theta \Delta}\right\rceil\right) \notag \\
&= \sum_{i=1+\left \lceil \frac{1}{\theta} \right \rceil }^{\infty}  \mathcal{I}\left(\left\lceil i\theta\right\rceil \eq \left\lceil (i-1)\theta\right\rceil\right) \int_{\frac{1}{\theta i}}^{\frac{1}{\theta (i-1)} } d \Delta \cdot \Delta \left(\left\lceil i\theta \right\rceil+ \Delta \left\lceil i\theta \right\rceil-\left\lceil i\theta\right\rceil^2 \Delta \right) \left(i+\theta \Delta i-i^2 \theta\Delta \right) \notag \\
&+ \sum_{i=1+\left \lceil \frac{1}{\theta} \right \rceil }^{\infty}\mathcal{I}\left(\left\lceil i\theta\right\rceil\neq\left\lceil (i-1)\theta\right\rceil\right) \left(\int_{\frac{1}{\left\lceil \theta (i-1)\right\rceil}}^{\frac{1}{\theta (i-1)}} d \Delta \cdot \Delta \left(\left\lceil (i-1)\theta \right\rceil+ \Delta \left\lceil (i-1)\theta \right\rceil-\left\lceil (i-1)\theta\right\rceil^2 \Delta \right) \left(i+\theta \Delta i-i^2 \theta\Delta\right) \right.\notag \\
& + \left.\int_{\frac{1}{\theta i}}^{\frac{1}{\left\lceil \theta (i-1) \right\rceil}} d \Delta \cdot \Delta \left(\left\lceil i\theta \right\rceil+ \Delta \left\lceil i\theta \right\rceil-\left\lceil i\theta\right\rceil^2 \Delta \right) \left(i+\theta \Delta i-i^2 \theta\Delta \right) \right) P_{ZR}\left(\mathcal{E}_1,\mathcal{E}_2,\left\lceil i\theta\right\rceil,i\right)  \notag \\
&+ \int_{1/(\theta \left\lceil \frac{1}{\theta} \right\rceil )}^{1} d \Delta \cdot 2\Delta \cdot \left(1-\Delta\right) \left( \left\lceil\frac{1}{\theta}\right\rceil+\theta \Delta \left\lceil\frac{1}{\theta}\right\rceil-\left\lceil\frac{1}{\theta}\right\rceil^2 \theta \Delta\right) P_{ZR}\left(\mathcal{E}_1,\mathcal{E}_2,2,\left\lceil\frac{1}{\theta}\right\rceil\right) \notag \\
&\overset{(a)}{=}\sum_{i=1+\left \lceil \frac{1}{\theta} \right\rceil}^{\infty} \left \{ \mathcal{I}\left(c(i) \eq c(i-1) \right) \left(\frac{c(i) (2i-1)}{2i \theta^2 (i-1)^2}+\frac{(3i^2-3i+1)c(i) (\theta(1-i)+1-c(i))}{3\theta^3 i^2(i-1)^3} \right) \right. \nonumber \\
&+ \mathcal{I}\left(c(i) \eq c(i-1)\right)\frac{(c(i) -1)c(i)(2i-1)(2i^2-2i+1)}{4\theta^3 (i-1)^3 i^3}  \nonumber \\
&\left. + \mathcal{I}\left(c(i)\neq c(i-1)\right) \frac{i \cdot c(i-1)}{2}\left( \frac{1}{\theta^2 (i-1)^2}-\frac{1}{c(i-1)^2}\right) \right. \nonumber \\
&\left. + \mathcal{I}\left(c(i)\neq c(i-1)\right) \frac{i \cdot c(i-1) (1-c(i-1)+\theta (1-i))}{3}\left(\frac{1}{\theta^3 (i-1)^3}-\frac{1}{c(i-1)^3}\right) \right. \nonumber \\
&\left.+ \mathcal{I}\left(c(i)\neq c(i-1)\right) \frac{i \cdot c(i-1) \cdot \theta (1-i)(1-c(i-1))}{4}\left( \frac{1}{\theta^4 (i-1)^4}-\frac{1}{c(i-1)^4}\right)\right.\nonumber \\
&\left.+ \mathcal{I}\left(c(i)\neq c(i-1) \right) \left[\frac{i \cdot c(i)}{2}\left(\frac{1}{c(i-1)^2}-\frac{1}{\theta^2 i^2}\right)+ \frac{i \cdot c(i)\cdot  (1-c(i)+\theta (1-i))}{3}\left( \frac{1}{c(i-1)^3}-\frac{1}{\theta^3 i^3}\right) \right] \right. \nonumber \\
&\left. +\mathcal{I}\left(c(i)\neq c(i-1)\right)\frac{i \cdot c(i) \cdot \theta (1-i)(1-c(i))}{4}\left(\frac{1}{c(i-1)^4}- \frac{1}{\theta^4 i^4}\right) \right \}P_{ZR}\left(\mathcal{E}_1,\mathcal{E}_2,\left\lceil i\theta\right\rceil,i\right)\notag \\
&+ \left \{\left( \left\lceil 1/\theta \right\rceil -\frac{1}{\theta^2 \left\lceil 1/\theta \right\rceil}\right) +\frac{2\left(\theta-\theta \left\lceil 1/\theta \right\rceil-1\right)}{3}\left(\left\lceil 1/\theta \right\rceil-\frac{1}{\theta^3 \left\lceil 1/\theta \right\rceil^2} \right)+\left(\left\lceil 1/\theta \right\rceil-\frac{1}{\theta^4 \left\lceil 1/\theta \right\rceil^3} \right)\frac{\theta(\left\lceil 1/\theta \right\rceil-1)}{2}\right \} \notag \\
&\;\;\;\;\;\;\;\;\;\;\;\;\;\;\;\;\;\;\;\;\;\;\;\;\;\;\;\;\;\;\;\;\;\;\;\;\;\;\;\;\;\;\;\;\;\;\;\;\;\;\;\;\;\;\;\;\;\;\;\;\;\;\;\;\;\;\;\;\;\;\;\;\;\;\;\;\;\;\;\;\;\;\;\;\;\;\;\;\;\;\;\;\;\;\;\;\;\;\;\;\;\;\;\;\;\;\;\;\;\;\;\;\;P_{ZR}\left(\mathcal{E}_1,\mathcal{E}_2,2,\left\lceil 1/\theta\right\rceil \right) \notag \\
&=C_1(\theta)
\end{align} In order to simplify the presentation, in step (a), we used the following change of variables $c(i)=\left\lceil i\theta\right\rceil$ and $c(i-1)=\left\lceil (i-1)\theta\right\rceil$.
Combining both sides of the inequality results in the lower bound given by (\ref{2nd}) in Section \ref{subsec:dual_zr_linear}. By analogy, $C_2(\theta)$ can be obtained the same way by swapping the roles of the two users.

\subsection{Divergence Bound- Upper Bounding $E_{sp}(R_1,R_2)$ for the Gaussian MAC \label{subsec:GMAC_Esp}}

Consider the Gaussian MAC defined in (\ref{eq:MAC_def}). 
For convenience, let us consider the subclass $\mathcal{W}$ of additive Gaussian MAC's
$Y\sim\mathcal{N}(x_1+x_2,\sigma_w^2)$. 
First, let us calculate the maximum
conditional mutual informations, $I(X_1;Y|X_2)$ and $I(X_2;Y|X_1)$.
\begin{eqnarray}
I(X_1;Y|X_2)&=&I(X_1;X_1+X_2+N|X_2)\nonumber\\
&=&h(X_1+X_2+N|X_2)-h(X_1+X_2+N|X_1,X_2)\nonumber\\
&=&h(X_1+N|X_2)-h(N)\nonumber\\
&\leq &\int_{-\infty}^{\infty}\mbox{d}x\cdot p_2(x)\cdot h(X_1+N|X_2=x)-\frac{1}{2}\log(2\pi e\sigma_w^2)\nonumber\\
&\leq &\int_{-\infty}^{\infty}\mbox{d}x\cdot p_2(x)\cdot\frac{1}{2}\ln [2\pi e\mbox{Var}\{X_1+N|X_2=x\}]-\frac{1}{2}\log(2\pi e\sigma_w^2)\nonumber\\
&\leq &\frac{1}{2}\ln [2\pi e\cdot\mathbf{E} \mbox{Var}\{X_1+N|X_2\}]-\frac{1}{2}\log(2\pi e\sigma_w^2)\nonumber\\
&=&\frac{1}{2}\ln [2\pi e\cdot \mbox{mmse}\{X_1+N|X_2\}]-\frac{1}{2}\log(2\pi e\sigma_w^2)\nonumber\\
&\leq &\frac{1}{2}\ln [2\pi e\cdot \mathbf{E}\{(X_1+N)^2\}]-\frac{1}{2}\log(2\pi e\sigma_w^2)\nonumber\\
&\leq&\frac{1}{2}\log\left(1+\frac{\mathcal{S}}{\sigma_w^2}\right).
\end{eqnarray}
Similarly, $I(X_2;Y|X_1)\leq \frac{1}{2}\log(1+\mathcal{S}/\sigma_w^2)$.
Both upper bounds are achieved at the same time if $X_1$ and $X_2$ are
independent, zero--mean, Gaussian random variables with variances $\mathcal{S}_1=\mathcal{S}_2=S$.
Thus, the conditions $R_1 \geq I(X_1;Y|X_2)$ and $R_2\geq I(X_2;Y|X_1)$, are equivalent to the condition
\begin{equation}
\label{constraint}
\sigma_w^2 \ge
\max\left\{\frac{\mathcal{S}}{e^{2R_1}-1},\frac{\mathcal{S}}{e^{2R_2}-1}\right\}\dfn \sigma_0^2(R_1,R_2),
\end{equation}
where $\sigma_0^2(R_1,R_2)$ is assumed larger than $\sigma^2$ since $(R_1,R_2)$ are assumed in the achievable region of the real underlying channel $P$. Now, 
\begin{equation}
\mathcal{D}(\mathcal{N}(x_1+x_2,\sigma_w^2)\|
\mathcal{N}(x_1+x_2,\sigma^2))=\frac{1}{2}\left[\frac{\sigma_w^2}{\sigma^2}-
\ln\left(\frac{\sigma_w^2}{\sigma^2}\right)-1\right],
\end{equation}
whose minimum under the constraint (\ref{constraint}) is
\begin{equation}
\mathcal{D}(\mathcal{N}(x_1+x_2,\sigma_0^2(R_1,R_2))\| \mathcal{N}(x_1+x_2,\sigma^2))=\frac{1}{2}\left[\frac{\sigma_0^2(R_1,R_2)}{\sigma^2}- \ln\left(\frac{\sigma_0^2(R_1,R_2)}{\sigma^2}\right)-1\right].
\end{equation}
Since this is independent of $(x_1,x_2)$, the outer maximization over $Q$ degenerates, and the end result is 
\begin{equation}
E_{sp}(R_1,R_2)\leq \frac{1}{2}\left[\frac{\sigma_0^2(R_1,R_2)}{\sigma^2}- \ln\left(\frac{\sigma_0^2(R_1,R_2)}{\sigma^2}\right)-1\right]\dfn 
\bar{E}_{sp}(R_1,R_2)
\end{equation}
\subsection{Minimization of the error exponents for the divergence bound \label{subsec:expo_deriv}}

The minimization of the first exponent $F_1$ given by (\ref{eq:F_functions_div1}) can be written explicitly as  
\begin{equation}\label{eq:F1}
F_1 = \min_{R\geq 0}2R + \frac{1}{2}\left\{\frac{ \mathcal{S}}{e^{2R}-1} - \ln\frac{ \mathcal{S}}{e^{2R}-1}-1\right\} 
\end{equation} Taking the first derivative of the function above based on $R$ and equating to zero as follows
$$\frac{d}{dR}F_1(R) = 2+\frac{1}{2}\left\{\frac{-2 \mathcal{S}e^{2R}}{(e^{2R}-1)^2} + \frac{2e^{2R}}{e^{2R}-1}\right\}=0 $$ yields $3x^2 - ( \mathcal{S}+5)x + 2 =0, \mathrm{with\;\;\;} x = e^{2R}$. The rate value that minimizes the first error exponent is obtained as $$R_1^*= \frac{1}{2}(\log(\mathcal{S}+5+\sqrt{(\mathcal{S})^2+10 \mathcal{S}+1})-\log(6)).$$
Using $R_1^*$, we finally get 
\begin{equation}
F_1^* = \ln\left(\frac{ \mathcal{S}+5+\sqrt{ \mathcal{S}^2+10 \mathcal{S}+1}}{6}\right)
+ \frac{1}{2}\left[\frac{6 \mathcal{S}}{ \mathcal{S}-1\sqrt{\mathcal{S}^2+10 \mathcal{S}+1}}-\ln\left(\frac{6 \mathcal{S}}{ \mathcal{S}-1\sqrt{\mathcal{S}^2+10 \mathcal{S}+1}}\right)-1\right].
\end{equation}
There is no difference in the minimization the second exponent $F_2(\alpha)$ apart from the role of $\alpha$. The minimum of $F_2(\alpha)$ is given by
\begin{equation}\label{eq:F2}
F_2(\alpha)^* = F_1^*- 2\alpha 
\end{equation} where $R_2^*=\frac{1}{2}(\log(\mathcal{S}+5+\sqrt{\mathcal{S}^2+10 \mathcal{S}+1})-\log(6))$.
Lastly, for the last exponent $F_{12}(\alpha)$ we have the following minimization based on $R$
for simplification we use the following change of variables $R'\triangleq R+\frac{\alpha}{2}$. Using our new variable $R'$ the minimization becomes
$F_{12}(\alpha) = \min_{R'\geq\frac{\alpha}{2}}2R'-\alpha+\frac{1}{2}\left\{\frac{2 \mathcal{S}}{e^{4R'}-1} - \ln\frac{ 2 \mathcal{S}}{e^{4R'}-1}-1\right\}$. Taking the first derivative of the third exponent and equating to zero
$$\frac{d}{dR'}F_{12}(R') = 2+\frac{1}{2}\left\{\frac{-8  \mathcal{S} e^{4R'}}{(e^{4R'}-1)^2} + \frac{4e^{4R'}}{e^{4R'}-1}\right\}$$ we get 
$2x^2 - (2 \mathcal{S}+3)x + 1 =0, \mathrm{with\;\;\;} x = e^{4R'}$. $R_{12}^*$ denotes the root of this equality which gives the minimum for the last exponent as follows.
\begin{equation}
 F_{12}^*(\alpha) = \frac{1}{2}\ln\left(\frac{ 2 \mathcal{S}+3+\sqrt{( 2 \mathcal{S})^2+12 \mathcal{S}+1}}{4}\right) +\frac{4 \mathcal{S}}{2 \mathcal{S}-1+\sqrt{(2 \mathcal{S})^2+12 \mathcal{S}+1}} 
+\frac{1}{2}\ln\frac{8 \mathcal{S}}{2 \mathcal{S}-1+\sqrt{(2 \mathcal{S})^2+12 \mathcal{S}+1}}-\alpha
\end{equation}


\bibliography{paper_IT}

\begin{thebibliography}{10}
\providecommand{\url}[1]{#1}
\csname url@samestyle\endcsname
\providecommand{\newblock}{\relax}
\providecommand{\bibinfo}[2]{#2}
\providecommand{\BIBentrySTDinterwordspacing}{\spaceskip=0pt\relax}
\providecommand{\BIBentryALTinterwordstretchfactor}{4}
\providecommand{\BIBentryALTinterwordspacing}{\spaceskip=\fontdimen2\font plus
\BIBentryALTinterwordstretchfactor\fontdimen3\font minus
  \fontdimen4\font\relax}
\providecommand{\BIBforeignlanguage}[2]{{%
\expandafter\ifx\csname l@#1\endcsname\relax
\typeout{** WARNING: IEEEtran.bst: No hyphenation pattern has been}%
\typeout{** loaded for the language `#1'. Using the pattern for}%
\typeout{** the default language instead.}%
\else
\language=\csname l@#1\endcsname
\fi
#2}}
\providecommand{\BIBdecl}{\relax}
\BIBdecl

\bibitem{goblick}
T.~Goblick, ``Theoretical limitations on the transmission of data from analog
  sources,'' \emph{IEEE Transactions on Information Theory}, vol.~11, pp.
  558--567, October 1965.

\bibitem{WynerZiv69}
A.~Wyner and J.~Ziv, ``On communication of analog data from a bounded source
  space,'' \emph{The Bell System Technical Journal}, vol.~48, pp. 3139--3172,
  Dec 1969.

\bibitem{WozencraftJacobs}
J.~Wozencraft and I.~M. Jacobs, \emph{Principles of Communication
  Engineering}.\hskip 1em plus 0.5em minus 0.4em\relax Wiley, New York, 1965.

\bibitem{Merhav12b}
N.~Merhav, ``On optimum parameter modulation-estimation from a large deviations
  perspective,'' \emph{IEEE Transactions on Information Theory}, vol.~58, pp.
  7215--7225, December 2012.

\bibitem{ZivZakai}
J.~Ziv and M.~Zakai, ``Some lower bounds on signal parameter estimation,''
  \emph{IEEE Transactions on Information Theory}, vol.~15, pp. 386--391,
  November 1969.

\bibitem{Cohn}
D.~Cohn, ``Minimum mean square error without coding,'' Ph.D. dissertation, MIT,
  June 1970.

\bibitem{Burnashev79}
M.~V. Burnashev, ``On the minimax detection of an inaccurately known signal in
  a white {Gaussian} noise background,'' \emph{Theory of Probability and Its
  Applications}, vol.~24, no.~1, pp. 107--119, 1979.

\bibitem{Burnashev84}
------, ``A new lower bound for the $\alpha$-mean error of parameter
  transmission over white gaussian channel,'' \emph{IEEE Transactions on
  Information Theory}, vol.~30, pp. 23--34, January 1984.

\bibitem{Burnashev85}
------, ``On minimum attainable mean-square error in transmission of a
  parameter over a channel with white {G}aussian noise,'' \emph{Problems of
  Information Transmission}, vol.~21, pp. 3--16, 1985.

\bibitem{Unsal-CISS}
A.~Unsal and R.~Knopp, ``{T}ransmission of correlated {G}aussian samples over a
  {M}ultiple-{A}ccess {C}hannel,'' in \emph{{CISS}2014, {IEEE} {C}onference on
  {I}nformation {S}ciences and {S}ystems, {M}arch 19-21, 2014, {P}rinceton,
  {NJ}}, 03 2014.

\bibitem{Unsal-thesis}
A.~Unsal, ``Transmission of analog source samples for remote and distributed
  sensing,'' Ph.D. dissertation, Telecom ParisTech, Nov. 2014.

\bibitem{Shannon59}
C.~E. Shannon, ``Probability of error for optimal codes in a {G}aussian
  channel,'' \emph{The Bell System Technical Journal}, vol.~38, pp. 611--656,
  May 1959.

\bibitem{Nazari-thesis}
A.~Nazari, \emph{{E}rror {E}xponent for {D}iscrete {M}emoryless
  {M}ultiple-{A}ccess {C}hannels}.\hskip 1em plus 0.5em minus 0.4em\relax The
  University of Michigan, Dec. 2011.

\bibitem{Ashikhmin}
A.~E. Ashikhmin, A.~Barg, and S.~N. Litsyn, ``A new upper bound on the
  reliability function of the {G}aussian channel,'' \emph{IEEE Transactions on
  Information Theory}, vol.~46, pp. 1945--1961, September 2000.

\bibitem{Weng}
L.~Weng, S.~Pradhan, and A.~Anastasopoulos, ``Error exponent regions for
  {G}aussian broadcast and multiple-access channels,'' \emph{IEEE Transactions
  on Information Theory}, vol.~54, pp. 2919--2942, July 2008.

\bibitem{Polyanskiy}
Y.~Polyanskiy, H.~Poor, and S.~Verd\'u, ``Channel coding rate in the finite
  blocklength regime,'' \emph{IEEE Transactions on Information Theory},
  vol.~56, pp. 2307--2359, December 2010.

\bibitem{Burnashev}
M.~V. Burnashev, ``{O}n relation between code geometry and decoding error
  probability,'' in \emph{{ISIT}2001, {IEEE} {I}nternational {S}ymposium on
  {I}nformation {T}heory, {J}une 24-29, 2001, {W}ashington, {DC}}, 06 2001.

\bibitem{Litsyn}
Y.~Ben-Haim and S.~Litsyn, ``Improved lower bounds on the reliability function
  of the {G}aussian channel,'' \emph{IEEE Transactions on Information Theory},
  vol.~54, pp. 5--12, January 2008.

\bibitem{Laplace}
P.~S. Laplace, ``M\'emoire sur la probabilit\'e des causes par les
  \'ev\`enements,'' \emph{M\'emoire de Math\'ematique et de Physique}, pp.
  621--656, 1774.

\bibitem{Massart}
B.~Massart and P.~Laurent, ``Adaptive estimation of a quadratic functional by
  model selection,'' \emph{Annals of Statistics}, vol.~28, pp. 1302--1338,
  2000.

\end{thebibliography}
\bibliographystyle{IEEEtran}


\end{document}